\documentclass[sigconf, nonacm]{acmart}
\usepackage{marginnote}

\usepackage[nameinlink, noabbrev, capitalize]{cleveref}
\usepackage{thm-restate}
\usepackage{booktabs} 
\usepackage{color,soul}
\usepackage{fancyvrb}
\usepackage{amsmath}
\usepackage{mathtools}
\usepackage{float}
\usepackage{bbm}
\usepackage{tcolorbox}
\usepackage{graphicx}
\usepackage{lipsum}
\usepackage{xfrac}
\usepackage{colortbl}
\usepackage{multirow}
\usepackage{balance}
\usepackage{paralist}
\usepackage{tabularx}
\usepackage{etoolbox}
\usepackage{caption}
\usepackage{subcaption}
\usepackage{algorithm}
\usepackage{algpseudocode}
\usepackage{enumitem}
\usepackage{thmtools, thm-restate} 
\usepackage{xspace}
\usepackage{tikz}
\usepackage{array}
\usepackage{xcolor}
\usepackage{multicol}
\usepackage{subcaption}
\usepackage{enumitem}
\usepackage{makecell}
\usepackage{tikz}
\usepackage{booktabs}
\usepackage{subcaption}
\usepackage[table]{xcolor}

{\typeout{#1}}
\newcommand\plabel[1]{\phantomsection\label{#1} } 

\newtheorem{theorem}{Theorem}[section]

\theoremstyle{remark}
\newtheorem{remark}[theorem]{Remark}
\newtheorem{example}[theorem]{Example}

\newtheorem*{problem*}{Problem}

\algrenewcommand\algorithmicensure{\textbf{Output:}}
\algrenewcommand\algorithmicrequire{\textbf{Input:}}

\def\OPT{\ensuremath{\mathrm{OPT}}}

\def\R{\ensuremath{\mathbb{R}}}

\def\cA{\ensuremath{\mathcal{A}}}
\def\cP{\ensuremath{\mathcal{P}}}
\def\cD{\ensuremath{\mathcal{D}}}
\def\cM{\ensuremath{\mathcal{M}}}

\def\cR{\ensuremath{\mathcal{R}}}

\def\cG{\ensuremath{\mathcal{G}}}

\def\cV{\ensuremath{\mathcal{V}}}
\def\cI{\ensuremath{\mathcal{I}}}

\def\Exp{\mathbb{E}}

\renewcommand{\Pr}{\mathop{\bf Pr\/}}

\newcommand{\set}[1]{\ensuremath{\{ #1 \}}}

\newcommand{\eps}{\varepsilon}

\newcommand{\eqdef}{\triangleq}

\newcommand{\norm}[1]{\left\|#1\right\|}
\newcommand{\abs}[1]{\left\vert#1\right\vert}

\newcommand{\E}[1]{\underset{#1}{\mathbb{E}}}

\newcommand{\nbr}{{\sim}}

\newcommand{\paratitle}[1]{\vspace{1mm}\noindent\textbf{{#1}.}}
\newcommand{\paren}[1]{\left( #1 \right)}

\newcommand{\Esymb}{\mathbb{E}}

\makeatletter

\def\Ex#1{%
    \ProbabilityRender{\Esymb}{#1}%
}

\def\ProbabilityRender#1#2{
  \@ifnextchar\bgroup%
  {\renderwithdist{#1}{#2}}
   {\singlervrender{#1}{#2}}
}
\def\singlervrender#1#2{%
   \ensuremath{\mathchoice
       {{#1}\left[ #2 \right]}
       {{#1}[ #2 ]}
       {{#1}[ #2 ]}
       {{#1}[ #2 ]}
   }
}
\def\renderwithdist#1#2#3{%
   \@ifnextchar\bgroup
   {\superfancyrender{#1}{#2}{#3}}
   {\ensuremath{\mathchoice
      {\underset{#2}{#1}\left[ #3 \right]}
      {{#1}_{#2}[ #3 ]}
      {{#1}_{#2}[ #3 ]}
      {{#1}_{#2}[ #3 ]}
     }
   }
}
\def\superfancyrender#1#2#3#4#5{
   \ensuremath{\mathchoice
      {\underset{#1}{{#1}}\left#4 #3 \right#5}
      {{#1}_{#2}#4 #3 #5}
      {{#1}_{#2}#4 #3 #5}
      {{#1}_{#2}#4 #3 #5}
   }
}
\makeatother

\newcommand{\rk}{k}

\newcommand{\MSD}{{MSD}\xspace}

\newcommand{\Greedy}{{\texttt{Greedy}}\xspace}
\newcommand{\DPG}{{\texttt{DP-Greedy}}\xspace}
\newcommand{\DPOSG}{{\texttt{DP-OSG}}\xspace}
\newcommand{\DPNOSG}{{\texttt{DP-NOSG}}\xspace}
\newcommand{\LS}{{\texttt{LS}}\xspace}
\newcommand{\DPSLS}{{\texttt{DP-SLS}}\xspace}
\newcommand{\Random}{{\texttt{Random}}\xspace}
\newcommand{\ExpMech}{{\texttt{EM}}\xspace}

\newcommand{\DPGRel}{{\texttt{DP-RelG}}\xspace}
\newcommand{\DPMix}{{\texttt{DP-Mix}}\xspace}

\newif\ifpaper
\paperfalse
\begin{document}

\title{Fast and Private Max-Sum Diversification}
\author{Ron Zadicario}
\affiliation{%
  \institution{Tel Aviv University}
  }  
\email{ronzadicario@mail.tau.ac.il}

\author{Tova Milo}
\affiliation{%
  \institution{Tel Aviv University}
  }
\email{milo@post.tau.ac.il}

\begin{abstract}
Result diversification is crucial for generating informative, non-redundant data summaries and query outputs. 
Although its various formulations have been extensively studied across an array of  data-driven disciplines, existing methods  fail to address the privacy concerns that arise when the underlying data is sensitive.
In this work, we initiate the study of result diversification under \emph{differential privacy}, focusing on the \emph{max-sum diversification} (\MSD) problem, a widely adopted model with the objective of maximizing a linear combination of a submodular function, quantifying relevance, and the sum of pairwise distances between selected items, quantifying diversity. We propose differentially private algorithms for \MSD under both cardinality and matroid constraints, achieving nearly optimal utility guarantees. At the same time, we design more efficient algorithms that maintain strong guarantees. Notably, the proposed algorithms are faster than existing non-private methods, making them appealing even in non-private settings. Experimental evaluations on real-world datasets demonstrate that the proposed approach achieves utility comparable to that of non-private baselines even under strong privacy guarantees, and significantly improves execution times for cardinality constraints.
\end{abstract}

\maketitle

\section{Introduction}
Result diversification is the problem of selecting a subset of items that is both relevant and diverse, subject to feasibility constraints. It underlies numerous data summarization and selection tasks, where capturing data variety, avoiding redundancy, and providing informative outputs are crucial. Consequently, its various formulations have been extensively studied across multiple domains, including databases, machine learning, and operations research.
 Among these formulations, the \emph{max-sum diversification} (\MSD) problem is one of the most well-studied and widely adopted diversity models in the database literature. It has been applied to a range of tasks, including query result diversification~\cite{deng2014complexity,borodin2017max,fraternali2012top,wang2018rc,agarwal2024computing}, graph summarization~\cite{song2018mining,guo2020diversified}, rule mining~\cite{bian2024discovering,han2026fast}, spatial keyword search~\cite{zhang2014diversified}, as well as recommender systems~\cite{castells2021novelty,puthiya2016coverage,coppolillo2024relevance} and feature selection~\cite{zadeh2017scalable,ghadiri2019distributed}.
 
In this problem, we are given a set of items, referred to as the \emph{ground set}, together with pairwise distances and a monotone submodular set function $f$ that quantifies the \emph{relevance} of each subset. The objective is to select a subset $S$ that maximizes a linear combination of $f(S)$ and the sum of pairwise distances among the items in $S$, subject to given feasibility constraints~\cite{borodin2017max}. Prior work has primarily focused on cardinality constraints, which limit the size of the selected subset, and more generally on \emph{matroid} constraints (defined in \Cref{sec:prelim}), which capture a broad range of practical feasibility criteria. A notable special case of this model is that of diverse Top-$k$ queries~\cite{gollapudi2009axiomatic}, where each item $v$ is associated with a score $q(v)$, the relevance of a subset $S$ is given by $f(S)=\sum_{v \in S} q(v)$, and the constraint is cardinality $|S|\le k$. By allowing arbitrary monotone submodular relevance functions and  matroid constraints, \MSD supports a substantially more expressive class of data summarization tasks, as illustrated in \Cref{example:motivating}.

Yet, many compelling applications of result diversification, and of \MSD in particular, involve sensitive individual data, such as purchase histories, locations, and search patterns. Improper handling of such data can pose significant privacy risks~\cite{zhang2014privacy,weinsberg2012blurme}, and therefore the maximization procedure must provide formal privacy guarantees for the individuals contributing to the dataset.

\begin{example}\label{example:motivating}
An online retailer wishes to select a set $S$ of $k$ popular products from a dataset filtered by $\texttt{category}=\text{``Health Care''}$, for purposes such as catalog design, promotion, or trend analysis. Selecting the Top-$k$ most purchased items may primarily reflect the preferences of a narrow group of highly active users. To instead maximize user reach, the retailer can maximize the coverage function $f(S)=\frac{1}{m}\left|\bigcup_{v\in S}\mathcal{U}(v)\right|$, where $\mathcal{U}(v)$ denotes the set of users who purchased item $v$ and $m$ is the total number of users. This objective is monotone submodular and often used for summarization~\cite{dasgupta2013summarization}. However, optimizing coverage alone may still yield a redundant set of items: \Cref{tab:motivating_redundant} shows results on the Amazon Reviews dataset~\cite{amazon}, where three of the six selected items (indicated by shaded rows) are essential oils. Selecting overly similar products may neglect users’ varied interests and reduce long-term satisfaction~\cite{Castells2015}.

To encourage diversity, each product is associated with a set of subcategories, and pairwise distances are defined using the Jaccard distance~\cite{leskovec2020mining} between these sets. The retailer can then maximize a linear combination of $f(S)$ and the sum of pairwise distances among items in $S$, producing a relevant and diverse set of products. Additional constraints can be incorporated naturally. For example, limiting the number of selected products within each price category induces a matroid constraint.

Critically, the selection procedure must also preserve users' privacy. Since the objective is derived from individual transaction logs, an exact solution could inadvertently reveal sensitive information. For instance, the inclusion of a niche product could signal the presence of a user with a rare medical condition. Preventing such leakage necessitates a mechanism that provides formal privacy guarantees. Using the algorithms proposed in this work, the retailer can generate a relevant, diverse, and privacy-preserving summary, as illustrated in \Cref{tab:motivating_diverse}, where the shaded rows indicate two products added by the private and diversity-aware algorithm.
\end{example}

\begin{figure}
    \centering
    \setlength{\tabcolsep}{2pt} 
  
    \begin{subtable}[t]{.485\columnwidth}
        \centering
        \fontsize{7pt}{8pt}\selectfont
        \begin{tabular}{ll}
            \toprule
            \textbf{Product} & \textbf{Subcategory} \\ \midrule
            Muscle Cream & Treatments \\
            \rowcolor{gray!15}Cassia Oil &  Essential Oils \\
            Oil Diffuser & Diffusers \\
            Heating Pad & Hot \& Cold Therapies \\
            \rowcolor{gray!15} Garlic Oil &  Essential Oils \\
            \rowcolor{gray!15} Rosemary Oil &  Essential Oils \\ 
            \midrule
            \multicolumn{2}{c}{\textbf{Coverage:} 0.15 \quad \textbf{Diversity:} 0.20} \\
            \bottomrule
        \end{tabular}
       \caption{Relevance-only selection. Three of six items are from the ``Essential Oils'' subcategory (shaded).}
        \label{tab:motivating_redundant}
    \end{subtable}
    \hfill
    \begin{subtable}[t]{.485\columnwidth}
        \centering
        \fontsize{7pt}{8pt}\selectfont
        \begin{tabular}{ll}
            \toprule
            \textbf{Product} & \textbf{Subcategory} \\ \midrule
            Muscle Cream & Treatments \\
            Cassia Oil & Essential Oils \\
            Oil Diffuser & Diffusers \\
            Heating Pad & Hot \& Cold Therapies \\
            \rowcolor{blue!10} Shoe Insoles & Foot Health \\
            \rowcolor{blue!10} Ear Otoscope & Ear Care \\ 
            \midrule
            \multicolumn{2}{c}{\textbf{Coverage:} 0.145 \quad \textbf{Diversity:} 0.31} \\
            \bottomrule
        \end{tabular}
        \caption{Selection balancing relevance and diversity. Shaded items are newly added.}
        \label{tab:motivating_diverse}
    \end{subtable}
    \caption{Representative products from the Amazon Health Care category. Product and subcategory names are simplified for readability.}
    \label{tab:motivating_example}
\end{figure}

Differential Privacy~\cite{dwork2006differential,dwork2014algorithmic} is a rigorous notion that has become the gold standard for analyzing sensitive data under strong privacy guarantees, and has been adopted by multiple companies~\cite{Abowd18,dwork2019differential} and governmental organizations~\cite{erlingsson2014rappor,ding2017collecting,tang2017privacy}.
Intuitively, the output distribution of a differentially private (DP) algorithm changes only minimally when a single individual’s data is modified, thereby ensuring that individual-level information is not disclosed (see \Cref{sec:prelim} for the formal definition).  Privacy is typically achieved by injecting noise into the computation process, which may result in some degradation in utility. In our setting, achieving stronger privacy often means computing a set of items with a slightly lower quality score.
Although DP submodular maximization has received considerable attention in recent years~\cite{gupta2010differentially,mitrovic2017differentially,chaturvedi2021differentially,rafiey2020fast}, the design of DP algorithms for the more general \MSD problem has yet to be addressed. Existing techniques for DP submodular maximization do not directly apply to \MSD,  as its objective function is not submodular.

In this work, we study the \MSD problem under differential privacy. We  propose DP algorithms that provide near-optimal utility guarantees under both cardinality and general matroid constraints. Our proposed algorithms also improve time complexity compared to existing non-private methods, making them valuable even in non-private settings. Experimental evaluation on two real-world applications, Amazon product and Uber pickup location summarization, shows that our approach achieves utility comparable to non-private baselines while significantly reducing the number of objective evaluations. We now state our contributions, with a summary of the main results in \Cref{tab:results}. To provide context, \Cref{sec:related} reviews relevant lower bounds from prior work.

\subsection{Our Results}\label{subsec:our_results}
 In the DP setting, the \MSD objective function depends on a sensitive dataset $D$ and is a linear combination of a submodular relevance term and a diversity term. It is called \emph{decomposable} if it can be expressed as a sum of functions, each 
 depending on a single record of $D$. The formal definitions are provided in \Cref{subsec:problem_formulation}.

We begin with a direct DP adaptation of the greedy algorithm of \citet{borodin2017max}, which serves as a natural baseline. Building on this, our main contribution is a faster DP algorithm whose analysis goes well beyond the original work, simultaneously addressing subsampling, privacy, and the absence of submodularity.

\begin{table*}[t]
\caption{Expected utility guarantees of $(\eps, \delta)$-DP algorithms and their query complexity. Here $\Delta$ is the sensitivity (equivalent to the decomposability parameter for decomposable objectives), $n$ is the ground set size, $\rk$ is the maximal feasible solution size, and $\gamma$ controls the utility–complexity trade-off. For general $\Delta$-sensitive functions with cardinality constraints, additive errors incur an additional  $\sqrt{\rk/\log(1/\delta)}$ factor. Setting $\eps=\infty$ yields non-private algorithms with the same multiplicative approximation ratios. The $1/2$-approximation ratio is tight under standard complexity-theoretic assumptions (see~\Cref{sec:related}).}
\footnotesize
\centering
\renewcommand{\arraystretch}{1.4} 
\setlength{\tabcolsep}{8pt}
\begin{tabular}{@{}llclll@{}}
\toprule
\textbf{\makecell[t]{Constraint}} & \textbf{\makecell[t]{Approx.}} & \textbf{\makecell[t]{Additive Error \\ {\scriptsize ($\Delta$-decomposable)}}} & \textbf{\makecell[t]{Oracle Calls}} & \textbf{\makecell[t]{Algorithm}} & \\
\midrule
$k$-Cardinality & $1/2$ & $O(\Delta k \eps^{-1} \log \delta^{-1} \log n)$ & $O(nk)$ & \DPG (\Cref{alg:greedy}) & (i) \\
\addlinespace[2pt]
& $1/2 - \gamma$ & $O(\Delta k \eps^{-1} \log \delta^{-1} \log n)$ & $O(n\log k\log \gamma^{-1})$ &  \DPNOSG (\Cref{alg:sample_greedy}) & (ii) \\
\addlinespace[2pt]
& $1 - 2/e - \gamma -1/\rk$ & $O(\Delta k \eps^{-1} \log \delta^{-1} \log n)$  & $O(n\log \gamma^{-1})$ & \DPOSG (\Cref{alg:sample_greedy}) & (iii) \\
\midrule
Matroid (rank $k$) & $1/2 - \gamma$ &    
 $O\left( \frac{\Delta \sqrt{k \log k \log \delta^{-1}}}{\eps\sqrt{\gamma}} \big( k \log n + \log \gamma^{-1} \big) \right)$
 & $O(\gamma^{-1}nk\log k)$ & \DPSLS (\Cref{alg:sample_local_search}) & (iv) \\
\bottomrule
\end{tabular}
\label{tab:results}
\end{table*}
\paratitle{Greedy algorithm for cardinality constraints}
Cardinality constraints, which capture the fundamental Top-$k$ query class, are among the most commonly studied constraint types. \citet{borodin2017max} proposed a \emph{non-oblivious}\footnote{An optimization algorithm is non-oblivious if it selects the next element with respect to an auxiliary function rather than the true objective.} greedy algorithm for cardinality-constrained \MSD, achieving a tight $1/2$-approximation. We extend their analysis to account for the error introduced by replacing greedy selections with privatized ones. For decomposable objectives, we show that the technique of~\citet{gupta2010differentially} provides improved privacy guarantees. The precise results, corresponding to row \textit{(i)} of \Cref{tab:results}, are discussed in \Cref{sec:greedy}. 

However, this algorithm makes $\Theta(nk)$ value oracle calls\footnote{Number of evaluations of the objective function; see \Cref{sec:prelim}.}, where $k$ is the cardinality bound and $n$ is the ground set size. For large $k$ (e.g., $k = \Omega(n)$), this complexity reaches $\Omega(n^2)$ oracle calls. Such scaling becomes a significant bottleneck in high-concurrency environments or interactive systems where each execution of the algorithm corresponds to a single user query and many such queries must be processed simultaneously. Even in well-resourced systems, the cumulative computational overhead of $\Theta(nk)$ evaluations per instance limits total system throughput and increases latency, underscoring the need for more efficient algorithms that better suit high-volume workloads.

\paratitle{Faster Greedy Algorithm via Subset Sampling}
Algorithms for  submodular maximization are often accelerated by ``sample greedy'' methods, where the next element is chosen from a smaller random subset rather than the entire ground set, improving efficiency at the cost of a slightly weaker approximation guarantee~\cite{buchbinder2017comparing,mirzasoleiman2015lazier,mitrovic2017differentially}. While these analyses rely on submodularity, which does not hold in our setting, we nonetheless propose an efficient ``sample greedy'' algorithm for the DP \MSD problem with cardinality constraints.
\emph{To our knowledge, this is the fastest known algorithm for \MSD, making it valuable even in non-private settings.}

We present two instantiations of the algorithm, each illustrating a distinct trade-off between complexity and utility. Furthermore, we extend the technique of \citet{gupta2010differentially} to incorporate the subsampling step and prove stronger guarantees for decomposable objectives.
Following established literature in both private \cite{mitrovic2017differentially} and non-private \cite{buchbinder2017comparing, mirzasoleiman2015lazier} submodular maximization, our utility guarantees are provided in expectation, as the randomness inherent in the subsampling process typically precludes the strong concentration required for meaningful high probability bounds.
The precise results, corresponding to rows \textit{(ii)} and \textit{(iii)} of \Cref{tab:results}, are detailed in \Cref{sec:sample_greedy}.

\paratitle{Local Search for Matroid Constraints}
Some data summarization tasks require more general constraints than simple cardinality, and matroids are among the  most studied such constraints.
Since the greedy paradigm does not provide any constant-factor approximation for the \MSD problem under matroid constraints, \citet{borodin2017max} showed that a natural single-swap local search algorithm achieves a $(1/2-\gamma)$-approximation. However, checking for local optimality or ensuring the selection of improving swaps is infeasible under differential privacy. To address this, we propose a DP local search algorithm for \MSD under matroid constraints that performs a predefined number of local search steps and selects swaps privately. 
The analysis bounds the loss incurred by allowing non-improving swaps, yielding near-optimal guarantees. We further improve efficiency by sampling a subset of candidate swaps in each iteration, reducing to $O(\gamma^{-1}\rk\log \rk)$ oracle calls, a substantial improvement over the $O(n^2 + \gamma^{-1}n\rk^2\log \rk)$ by \citet{borodin2017max}. \emph{To our knowledge, this algorithm achieves the best known complexity for \MSD under matroid constraints.}

Importantly, even without privacy, the algorithm achieves (in expectation) the same approximation ratio as \citet{borodin2017max} with improved complexity. The precise statements of our results, corresponding to row \textit{(iv)} of \Cref{tab:results}, are given in \Cref{sec:matroid_local_search}.

\section{Related Work}\label{sec:related}
We provide a brief survey of relevant literature, starting with non-private settings, followed by lower bounds, and finally covering additional privacy-related contributions.

\paratitle{Non-Private Result Diversification}
The \MSD problem generalizes the \emph{max-sum dispersion} problem (also called \emph{remote-clique}), which aims to maximize $\sum_{\{u,v\} \subseteq S} d(u,v)$ and has attracted sustained attention for over three decades~\cite{hansen1988dispersing, ravi1994heuristic, hassin1997approximation, 2003Fekete}. \citet{gollapudi2009axiomatic} first studied \MSD for modular relevance and cardinality constraints, obtaining a $1/2$-approximation by reduction to max-sum dispersion, which was later extended to matroid constraints by \citet{abbassi2013diversity}. \citet{borodin2017max} substantially generalized this framework to submodular relevance and matroid constraints, providing a greedy algorithm for cardinality constraints and a local search algorithm for general matroid constraints with $1/2$ and $(1/2-\gamma)$ approximations, respectively. For Euclidean distances with modular relevance, \citet{indyk2014composable} introduced composable coresets for streaming and distributed settings. \citet{ceccarello2017mapreduce} later generalized these constructions to metrics with bounded doubling dimension, and \citet{ceccarello2018fast} extended the framework to matroid constraints.  A fully dynamic variant for doubling metrics was studied by \citet{pellizzoni2025fully}.
\citet{agarwal2020efficient} proposed near-linear-time algorithms using efficient indexes for Euclidean distances. Additionally, \citet{cevallos2017local} studied \MSD under negative-type distances, a setting that is incomparable to general metric distances. Despite the extensive literature, none of these works provides formal privacy guarantees.

\paratitle{DP Submodular Maximization}
As the \MSD problem has not been studied in the context of differential privacy, we review the most closely related line of work, which focuses on DP submodular maximization.
The work of \citet{gupta2010differentially} was the first to address the problem, focusing on decomposable objective functions under cardinality constraints. \citet{mitrovic2017differentially} considered the more general case of submodular objectives with bounded sensitivity, proposing algorithms for both monotone and non-monotone functions under matroid and $p$-extendable systems constraints. 
\citet{chaturvedi2021differentially} revisited decomposable objectives, extending the results of \citet{gupta2010differentially} from cardinality to matroid constraints and from monotone to non-monotone functions.
\citet{rafiey2020fast} studied private submodular and $k$-submodular maximization under matroid constraints.
Other notable advances include \citet{salazar2021differentially}, who studied the problem in the online setting under cardinality constraints; \citet{sadeghi2021differentially}, who considered bounded-curvature objectives; \citet{zadicario2026differentially}, who addressed knapsack constraints for both monotone and non-monotone objectives; and \citet{chaturvedi2023streaming}, who proposed a streaming algorithm for cardinality constraints.
 Yet, all of these approaches apply only to submodular objectives and do not directly extend to \MSD.

\paratitle{Hardness and Lower Bounds}
Cardinality-constrained max-sum dispersion is known to be NP-hard even with metric distances~\cite{hansen1988dispersing}; hence, the more general \MSD problem is also NP-hard. Moreover, \citet{borodin2017max} showed that achieving an approximation factor better than $1/2$ is hard under the planted-clique hardness assumption~\cite{alon2011inapproximability}. As summarized in \Cref{tab:results}, the guarantees in rows (i), (ii), and (iv) match or nearly match this tight approximation factor. In contrast, row (iii) achieves improved complexity while providing a slightly weaker approximation guarantee.

Additionally, the DP \MSD problem in \Cref{subsec:problem_formulation} generalizes DP monotone submodular maximization, and hence known lower bounds for that problem extend to our setting. In particular, \citet{gupta2010differentially} proved that any $\eps$-DP algorithm for maximizing a submodular function under a cardinality constraint $\rk$ over a ground set of size $n$ must incur expected additive error $\Omega(\Delta \rk \eps^{-1} \log(n/\rk))$. \citet{chaturvedi2023streaming} extended this to the $(\eps,\delta)$-DP setting, showing that any such algorithm achieving multiplicative factor $c$ must incur additive error $\Omega(\Delta \rk \eps^{-1} c \log(\eps/\delta))$, assuming $n \ge \delta^{-1} \rk (e^\eps - 1)$ and $c \ge 4\delta/(e^\eps - 1)$.
In both cases, the hard instances are $1$-decomposable, implying that our results in rows (i)–(iii) of \Cref{tab:results} match the corresponding lower bounds up to logarithmic factors. These results recover the best-known additive error bounds for submodular maximization, and note that even in this restricted setting the tightness of logarithmic factors remains open~\cite{gupta2010differentially}.
For general bounded-sensitivity functions (row (iv), and rows (i)–(iii) with an additional $\sqrt{\rk/\log(1/\delta)}$ factor), our guarantees incur an additional $\sqrt{\rk}$-factor compared to the known lower bounds, which are not known to be tight in this more general setting. A similar $\tilde{O}(\sqrt{\rk})$ gap also appears in the state-of-the-art results for bounded-sensitivity submodular maximization under cardinality constraints~\cite{mitrovic2017differentially}.

\section{Preliminaries}
\label{sec:prelim}
In this section, we present the basic definitions and notation, introduce submodular set functions and matroids, and review  key concepts from differential privacy relevant to this work.

\paratitle{Notation and Set Functions} Let $V$ denote a finite ground set of size $n$. For a set function $f:2^V \to \mathbb{R}$, the \emph{marginal contribution} of an element $u\in V$ to a set $S\subseteq V$ is $f(u\mid S) = f(S\cup \{u\}) - f(S)$. A function $d: V\times V \to \mathbb{R}_{\ge 0}$ is a \emph{pseudometric} if it is symmetric, satisfies the triangle inequality, and $d(u,u)=0$ for all $u\in V$. Extend $d$ to sets via $d(S) = \sum_{\{u,v\} \subseteq S} d(u,v)$ and $d(S,T) = \sum_{u\in S, v\in T} d(u,v)$ for disjoint sets $S,T\subseteq V$. For $u\in V$, write $d(u\mid S) = d(S\cup\{u\}) - d(S) = \sum_{v\in S} d(u,v)$.

\paratitle{Submodular Functions and Matroids} A set function $f:2^V \to \mathbb{R}$ is \emph{monotone} if $f(S)\le f(T)$ for all $S\subseteq T$, \emph{non-negative} if $f(S)\ge 0$ for all $S\subseteq V$, and \emph{submodular}~\cite{nemhauser1978analysis} if $f(u\mid S) \ge f(u\mid T)$ for all $S\subseteq T$ and $u\in V\setminus T$. A \emph{matroid}~\cite{whitney1992abstract} is a pair $(V,\cI)$, where $\cI\subseteq 2^V$ satisfies \textit{(i)} if $A\subseteq B$ and $B\in \cI$ then $A\in \cI$, and \textit{(ii)} for $A,B\in\cI$ with $|A|<|B|$, there exists $u\in B\setminus A$ such that $A\cup\set{u}\in\cI$. Subsets in $\cI$ are called \emph{independent sets}, and inclusion-wise maximal independent sets are called \emph{bases}. All bases have the same cardinality, known as the \emph{rank} of the matroid.
See
\ifpaper
{Appendix~A in~\cite{full_version}}
\else
{\Cref{appendix:matroids}}
\fi
for additional background and examples of matroids.

\paratitle{Complexity Model}
We adopt the oracle complexity model, the standard abstraction in set function optimization, as it captures algorithmic complexity independently of application-specific details such as the cost of evaluating the objective function. Thus, we assume access to $f$, $d$, and $\cI$ via oracles: a \emph{value oracle} returns $f(S)$ or $d(S)$ for a given $S$, and an \emph{independence oracle} determines whether $S \in \cI$. Complexity is measured by the number of oracle calls. While each marginal gain $d(u\mid S)$ may require up to $k$ pairwise evaluations $d(u,v)$, we treat $d$ as a set function and count each evaluation of $d(S)$ as a single oracle call for consistency with the submodular objective. Although the number of oracle calls is generally a reliable proxy for runtime, discrepancies may arise in practice, as evaluating the oracle on larger sets can be slower in some applications.

\paratitle{Differential Privacy} 
A dataset $D$ is a finite multiset of records from a domain $\mathcal{X}$. The number of records in $D$ is denoted by $m = |D|$.  We let $\mathcal{D}$ denote the space of all possible datasets over this domain.
Two datasets $D$ and $D'$ are called \emph{neighboring} (denoted $D\nbr D'$) if they differ in one record.  Intuitively, differential privacy ensures that the distribution of outputs of a randomized algorithm  does not significantly change when the data of a single individual is changed.
\begin{definition}[Differential Privacy~\cite{dwork2006differential, dwork2014algorithmic}]\label{def:DP}
    A randomized algorithm $\cA$ is $(\eps,\delta)$-differentially private (DP) if for any neighboring datasets $D\nbr D'$  and any set of possible outputs $S\subseteq Range(\cA)$, 
    \[
       \Pr[\cA(D)\in S] \le e^\eps \cdot \Pr[\cA(D')\in S] +\delta.
    \]
    If $\delta=0$, we say that $\cA$ is $\eps$-DP.
\end{definition}
In our privacy analysis, we use the following standard composition results, which capture how privacy loss accumulates when multiple DP algorithms are applied sequentially. A more general statement appears in \cite{dwork2014algorithmic}.
\begin{theorem}[Composition~\cite{dwork2010boosting,dwork2014algorithmic}]\label{thm:composition}
Let $\cA_1,\dots,\cA_k$ be $\eps$-DP algorithms. Their $k$-fold adaptive composition $\cA_{[k]}$, which outputs $y_i=\cA_i(D, y_1,\dots,y_{i-1})$ for $i=1,\dots,k$, satisfies:
\begin{enumerate}[leftmargin=*,label=(\roman*)]
    \item  $k\eps$-DP (Basic),
    \item $(\sqrt{2k\log(1/\delta)}\eps + k\eps(e^\eps-1),\, \delta)$-DP for any $\delta>0$ (Advanced).
\end{enumerate}
In particular, to achieve $(\eps, \delta)$-DP for $\eps\in (0,1)$, it suffices that each $\cA_i$ satisfies $\frac{\eps}{2\sqrt{2k\log(1/\delta)}}$-DP.
\end{theorem}
DP algorithms are typically calibrated to the sensitivity of the underlying function with respect to single-record modifications of the dataset, defined formally as follows.
 
\begin{definition}[Sensitivity \cite{dwork2014algorithmic}]\label{def:sensitivity}
The \emph{sensitivity} of a function $q_D: \cR \to \R$ that depends on a dataset $D$ is defined as
$
\Delta_q = \sup_{r\in\cR}\sup_{D\nbr D'} |q_D(r)-q_{D'}(r)|.
$
In particular, the sensitivity of a set function $q_D: 2^V \to \R$ is
$\sup_{S \subseteq V}\sup_{D \sim D'} |q_D(S) - q_{D'}(S)|$.
\end{definition}
The exponential mechanism \cite{mcsherry2007mechanism} is a primitive for privately selecting an approximately highest-scoring element from a candidate set according to a quality function that depends on the sensitive dataset. 
\begin{definition}[The Exponential Mechanism \cite{mcsherry2007mechanism}]\label{def:exp-mech}
Given $D\in\cD$, a finite set of candidates $\cR$, a quality function  $q_D:\cR \to\R$, and a privacy parameter $\eps$, the exponential mechanism  $\ExpMech(D,q,\cR,\eps)$ outputs $r\in \cR$ with probability proportional to 
    $\exp\Big( \frac{\eps \cdot q_D(r)}{2\Delta_q} \Big)$.
\end{definition}
We assume that the sensitivity bound $\Delta_q$ is provided alongside $q$ and do not consider it as a separate input. 
\begin{theorem}
[\cite{mcsherry2007mechanism, bassily2016algorithmic}]\label{thm:exponential_mech}
The exponential mechanism is $\eps$-DP.
Moreover, 
\[
\Ex{\ExpMech(D,q,\cR,\eps)} \ge  \max_{r\in \cR}q_D(r) -  \frac{2\Delta_q \log\abs{\cR}}{\eps}.
\]
\end{theorem}
 The utility guarantee in expectation, which we use in this work, is due to \cite{bassily2016algorithmic}. The privacy guarantee is due to \cite{mcsherry2007mechanism}, as is the high-probability utility guarantee, which we omit here.

\section{Problem Formulation}\label{subsec:problem_formulation}
The DP max-sum diversification (\MSD) problem studied in this work extends the setting of \citet{borodin2017max} to incorporate differential privacy requirements, and is formalized below.

\begin{definition}[DP max-sum diversification]
An instance of DP \MSD is a tuple $\langle D, V, \cI, f_D, d_D \rangle$ where:
\begin{itemize}[leftmargin=*, label=--]
    \item $D$ is a sensitive dataset.
    \item $V$ is a (public) dataset of items, called the \emph{ground set}.
    \item $\cI \subseteq 2^V$ is a collection of feasible subsets.
    \item $f_D: 2^V \to \R^+$ is a monotone submodular function that depends on the sensitive dataset $D$.
    \item $d_D: V \times V \to \R$ is a pseudometric (possibly depending on $D$).
    and we write $d_D(S) \eqdef \sum_{\{u,v\} \subseteq S} d_D(u,v)$.
\end{itemize}
The goal is to find $S \in \cI$ that approximately maximizes  
\[
\phi_D(S) \eqdef  f_D(S) + \lambda \cdot d_D(S)
\] 
while satisfying differential privacy with respect to $D$.
\end{definition}
The parameter $\lambda \ge 0$ controls the trade-off between relevance and diversity. The subscript $D$ is omitted when it is clear from context.

\begin{remark}\label{remark:submodular}
The relevance function $f$ is submodular and the diversity term $d(S)$, being a sum of pairwise distances, is \emph{supermodular}. Hence, the objective $\phi$ is not necessarily submodular.
\end{remark}

\begin{example}\label{example:notation_example} 
In the setting of \Cref{example:motivating}, the sensitive dataset $D$ consists of purchase records where each entry specifies the products purchased by a user. The ground set $V$ contains Amazon products along with attributes such as product name, price, and category. The relevance function is defined as $f_D(S) = \frac{1}{m} \bigl|\bigcup_{v \in S} \mathcal{U}_D(v)\bigr|$, where $\mathcal{U}_D(v)$ is the set of users in $D$ who purchased item $v$. Let $C(v)$ denote the set of categories of item $v \in V$, and the diversity is quantified by the Jaccard distance $d(v, v') = 1 - \frac{|C(v) \cap C(v')|}{|C(v) \cup C(v')|}$ between category sets of products $v$ and $v'$.
\end{example}

\paratitle{Decomposable Objective Functions}
A set function $\phi_D$ is called \emph{$\Delta$-decomposable}~\cite{gupta2010differentially,mitrovic2017differentially} if it can be written as $\phi_D(S) = \sum_{x \in D} \phi_x(S)$, where each  $\phi_x: 2^V \to [0, \Delta]$ depends only on the individual record $x$. A $\Delta$-decomposable function has sensitivity $\Delta$.
Prior work has studied DP maximization of decomposable submodular functions~\cite{gupta2010differentially,chaturvedi2023streaming,chaturvedi2021differentially}, as they allow stronger privacy guarantees compared to the general case. Our algorithms for cardinality constraints similarly leverage this structure to provide privacy guarantees that nearly match known lower bounds for decomposable objectives.

\begin{example}\label{example:decomposable}
In  the setting of \Cref{example:motivating}, associate each user $x \in D$ with a monotone, submodular utility $f_x(S) = \min\{|S \cap \cP(x)|, 1\}$, where $\cP(x)$ is the set of products purchased by user $x$. This function indicates whether user $x$ is covered by $S$ and depends solely on the data of that individual. 
For $\lambda\in [0,1]$, define the per-user objective
\[
\phi_x(S) = \frac{1-\lambda}{m} f_x(S) + \frac{2\lambda}{mk(k-1)}d(S)\]
where $m=|D|$. The normalization ensures that $\phi_x(S) \in [0,1/m]$. The total objective is the average of these individual contributions: \[\phi(S) = \sum_{x \in D} \phi_x(S) = \frac{1-\lambda}{m} \sum_{x \in D} f_x(S) + \frac{2\lambda}{k(k-1)} d(S).\] Thus, $\phi$ is $1/m$-decomposable.
We now derive the sensitivity bound for $\phi$. If $D = (D' \setminus \{x'\}) \cup \{x\}$, then for any $S \subseteq V$,
$\abs{\phi_D(S) - \phi_{D'}(S)} = \abs{\phi_x(S) - \phi_{x'}(S)} \le 1/m$,
and thus $\phi$ has sensitivity bounded by $1/m$.
\end{example}
\section{Greedy Algorithm for Cardinality Constraints}
\label{sec:greedy}

We start by presenting \DPG, a DP adaptation of the greedy algorithm from \cite{borodin2017max}, replacing each selection step with the exponential mechanism. Although it achieves high utility, as also shown empirically in \Cref{sec: exp}, its complexity scales linearly in $k$, which can become prohibitive when $k$ is large. \emph{Addressing this limitation constitutes a primary contribution of this work} and is the focus of the subsequent section. \DPG, stated as  \Cref{alg:greedy}, is \emph{non-oblivious}: it selects elements according to the auxiliary function $\phi_D'(S) = \frac{1}{2} f_D(S) + \lambda d_D(S)$ rather than the objective $\phi_D$.

\begin{algorithm}
\caption{\DPG}
\label{alg:greedy}
\begin{algorithmic}[1]
\Require 
Dataset $D$,
ground set $V$,
monotone submodular $f_D$,
pseudometric $d_D$,
cardinality $\rk$,
 privacy parameter $\eps_0$,
diversity parameter $\lambda$
\Ensure Size $\rk$ subset of $V$
\State Define $\phi'_D(\cdot) = \frac{1}{2} f_D(\cdot) + \lambda d_D(\cdot)$
\State Set $S_0 \gets \emptyset$
\For{$i=1,\dots,\rk$}
\State Let $N_i \gets V\setminus S_{i-1}$
\State For all $u\in N_i,\,$ define $q^i_D(u) = \phi'_D(u\mid S_{i-1})$
\State Compute $u_i \gets \ExpMech(D,q^i, N_i,\eps_0)$ 
\State $S_i \gets S_{i-1}\cup\set{u_i}$
\EndFor
\end{algorithmic}
\end{algorithm}
\normalsize
\begin{theorem}\label{thm:greedy_guarantee}
Suppose $\phi_D$ has sensitivity $\Delta$. Given parameter $\eps_0 > 0$, \DPG  is $\rk\eps_0$-DP and $(\eps,\delta)$-DP for all $\delta>0$ with $\eps = \sqrt{2\rk \log(1/\delta)}\,\eps_0 + \rk \eps_0 (e^{\eps_0}-1)$. Moreover, it outputs a feasible set $S$ such that $\Exp[\phi_D(S)] \ge \frac{1}{2} \phi_D(\OPT) - O(\frac{\rk \Delta \log n}{\eps_0})$, and makes $O(n \rk)$ oracle  calls.
\end{theorem}

In the special case where $\phi_D$ is decomposable (and hence so is~$\phi'_D$), \DPG provides a stronger privacy guarantee  compared to the general case, as stated next.
The assumption that $\phi_D$ is $1$-decomposable is without loss of generality, since we can apply our algorithms to the normalized function $\phi_D / \Delta$, and the guarantees for $\phi_D$ incur an additive error scaled by $\Delta$.
\begin{theorem}\label{thm:greedy_decomposable}
    Suppose $\phi'_D$ is $1$-decomposable. Then given parameter $\eps_0\in(0,1]$ , \DPG (\Cref{alg:greedy}) is $(\eps,\delta)$-DP for every $\delta>0$ and $\eps=(e^{\eps_0/2}-1)(4+\log(1/\delta))$.
\end{theorem}

\Cref{thm:greedy_guarantee,thm:greedy_decomposable} correspond to row \textit{(i)} in \Cref{tab:results}.
Informally, advanced composition implies that setting $\varepsilon_0 = O(\varepsilon / \sqrt{\log(1/\delta)\rk})$  ensures $(\varepsilon, \delta)$-DP. 
 However, \Cref{thm:greedy_decomposable} implies that setting $\varepsilon_0 = O(\varepsilon/ \log(1/\delta))$ suffices to ensure $(\varepsilon, \delta)$-DP, avoiding the dependence on $\rk$ that arises from $\rk$ compositions in the general case. Substituting this into the utility guarantee of \Cref{thm:greedy_guarantee} yields the stated additive error in \Cref{tab:results}.
\Cref{thm:greedy_decomposable} is proved by the same steps as in the privacy proof for the CPP problem in~\cite{gupta2010differentially}. 
Since this result follows as a special case of our more general analysis
 in \Cref{sec:sample_greedy},
 it is omitted. In the remainder of this section, we prove~\Cref{thm:greedy_guarantee}.

Let $\OPT$ be an optimal solution, i.e.,  a subset of size at most $\rk$ that maximizes $\phi$. Since $\phi$ is monotone, we may assume that $\abs{\OPT}=\rk$. Let $S_i$ be the solution at the end of iteration $i$, and define
$C_i=\OPT\setminus S_{i-1}$. 
The approximation guarantee of the greedy algorithm in \cite{borodin2017max} relies on a lemma which relates the overall distance between the sets $S_{i-1}$ and $C_i$ to the value $d(\OPT)$. An adaptation of this lemma is stated next.
The proof is deferred to \Cref{appendix:greedy}.

\begin{restatable}[Adapted from \cite{borodin2017max}]{lemma}{distlowerbound}
\label{lemma:dist_lowerbound}
    For all $i\in [\rk] \eqdef \set{1,\dots, \rk}$,  $d(C_i,S_{i-1}) \ge \frac{(i-1)\abs{C_i}}{\rk(\rk-1)}d(\OPT)$.
\end{restatable}

\begin{proof}[Proof Sketch of \Cref{thm:greedy_guarantee}]
For query complexity, note that the algorithm performs $\rk$ iterations of $O(n)$ oracle  calls. The privacy follows from the $\eps_0$-DP of the exponential mechanism and composition (\Cref{thm:composition}). 
For utility, let $f'(S)=f(S)/2$. Consider iteration $i\in [\rk]$ and condition on $S_{i-1}$. The
submodularity and monotonicity of $f'$ imply that $\sum_{u\in C_i} f'(u\mid S_{i-1})\ge  f'(\OPT)-f'(S_{i-1})$.
Furthermore, by \Cref{lemma:dist_lowerbound}, we have  $\sum_{u\in C_i} d(u\mid S_{i-1}) \ge \frac{{(i-1)}\abs{C_i}}{\rk(\rk-1)}d(\OPT)$. 
Thus, the exponential mechanism, executed with the marginal gain quality functions $q^i_D$ of sensitivity $2\Delta$, guarantees that
\begin{align*}
    \Exp[\phi'(u_i\mid S_{i-1})] &\ge  \max_{u\in V\setminus S_{i-1}} \phi'(u\mid S_{i-1}) - \alpha  \ge \tfrac{\sum_{u\in C_i }\phi'(u \mid S_{i-1})}{\abs{C_i}} -\alpha  \\
    & \ge \tfrac{ f'(\OPT)-f'(S_{i-1})}{\rk} + \tfrac{\lambda(i-1)}{\rk(\rk-1)}d(\OPT) - \alpha 
\end{align*}
for $\alpha = O(\eps_0^{-1} \Delta \log n)$.
Monotonicity, together with first removing the conditioning on $S_{i-1}$
 and then taking expectation over all its realizations, yields
\[\Exp[\phi'(u_i | S_{i-1})] \ge \tfrac{ f'(\OPT)- \Exp[f'(S_{\rk})]}{\rk} + \tfrac{\lambda(i-1)}{\rk(\rk-1)}d(\OPT) - \alpha.\]
Summing over $i\in[\rk]$ gives \[Exp[\phi'(S_{\rk})] \ge f'(\OPT)- \Exp[f'(S_{\rk})] + \tfrac{\lambda}{2}d(\OPT) - \rk\alpha.\] Plugging the definitions for $f',\phi'$ and rearranging completes the proof.
\end{proof}
\section{Faster Greedy Algorithm via Subset Sampling}
\label{sec:sample_greedy}

In this section, we present our main contribution for the cardinality-constraint setting. The Sample Greedy algorithm for monotone submodular maximization~\citet{buchbinder2017comparing,mirzasoleiman2015lazier} achieves a \mbox{$(1-1/e-\gamma)$}-approximation in expectation using only $O(n\log\gamma^{-1})$ oracle calls. \citet{mitrovic2017differentially} proposed a variant for non-monotone submodular functions, guaranteeing a $(1-1/e)/e$-approximation, and adapted it to satisfy differential privacy. However, these analyses rely on submodularity, which does not hold in our setting. 

Nevertheless, we propose a Sample Greedy algorithm for DP \MSD under cardinality constraints, stated as  \Cref{alg:sample_greedy}. It achieves strong utility guarantees while using substantially fewer oracle calls than \DPG. To the best of our knowledge, this is the fastest known algorithm for \MSD with cardinality constraints, which is of interest even without privacy. Experimental results (\Cref{sec: exp}) show that our method significantly improves running time over \DPG while preserving utility. We further extend the analysis of \citet{gupta2010differentially} to account for the subsampling step and obtain improved privacy guarantees for decomposable objectives.

\begin{algorithm}
\caption{\DPNOSG (DP Non-Oblivious Sample Greedy) \;/\; \DPOSG (DP Oblivious Sample Greedy)}
\label{alg:sample_greedy}
\begin{algorithmic}[1]
\Require 
Dataset $D$,
ground set $V$,
submodular $f_D$,
pseudometric $d_D$,
cardinality $\rk$,
 privacy parameter $\eps_0$,
utility parameter~$\gamma$,
diversity parameter $\lambda$.
\Ensure Size $\rk$ subset of $V$

\State In \DPNOSG:
\Statex \quad  $\phi'_D(\cdot ) = \frac{1}{2-\gamma} f_D(\cdot)+ \lambda d_D(\cdot)$,\,  $g(i)=\rk-i+1$ for all $i\in[\rk]$
\State In \DPOSG:
\Statex \quad  $\phi'_D(\cdot) = f_D(\cdot)  + \lambda d_D(\cdot) $,\,  $g(i) = \min\{\rk, n-i+1\}$ for all $i\in[\rk]$
\State Set $S_0 \gets \emptyset$
\For{$i=1,\dots,\rk$}
         \State Set $N_i\gets V\setminus S_{i-1}$ 
        \State Sample a subset $V_i\subseteq N_i$ of size $\left\lceil\abs{N_i}\cdot \min\left\{\frac{\log(1/\gamma)}{g(i)}, 1\right\}\right\rceil$  uniformly at random. \plabel{cardinality_sampling}
        \State For all $u\in V_i,\,$ define $q^i_D(u)= \phi'_D(u\mid S_{i-1})$
\State Compute $u_i \gets \ExpMech(D,q^i_D, V_i, \eps_0)$
\State $S_i \gets S_{i-1}\cup\set{u_i}$
\EndFor
\State \Return $S_\rk$
\end{algorithmic}
\end{algorithm}
\normalsize

\Cref{alg:sample_greedy} performs greedy selections with respect to a function $\phi'$ that can be either the objective $\phi$ or an auxiliary function, and a utility function $g \colon [\rk] \to \mathbb{N}$, where $g(i)$ determines the size of the sampled candidate set in iteration~$i$.
We study two instantiations corresponding to different choices of $g$ and $\phi'$, yielding different trade-offs between running time and approximation guarantees. 

\subsection{Non-Oblivious Sample Greedy}\label{subsec:non-oblivious-sample}
The first instantiation we consider, referred to as \DPNOSG,  selects the next element according to the auxiliary function $\phi'_D:2^V\to\R$ given by  
$
    \phi'_D(S) = \frac{1}{2 -\gamma}\cdot f_D(S)+ \lambda d_D(S), 
$
and uses the utility function $g(i)=\rk-i+1$ for all $i\in[\rk]$. 
We obtain the following result.
\begin{theorem}\label{thm:sample_greedy_guarantee_non_oblivious}
Suppose $\phi_D$ has sensitivity $\Delta$. Then,  \DPNOSG (\Cref{alg:sample_greedy}) with parameters $\eps_0 > 0$ and $\gamma \in (0,1)$ is $\rk\eps_0$-DP and $(\eps,\delta)$-DP for all $\delta > 0$, where $\eps = \sqrt{2\rk\log(1/\delta)}\,\eps_0 + \rk\eps_0 (e^{\eps_0}-1)$. Moreover, it outputs a feasible set $S$ such that $\Exp[\phi_D(S)] \ge (\tfrac{1}{2}- \gamma) \cdot \phi_D(\OPT) - O\Big(\frac{\rk\Delta \log n}{\eps_0}\Big)$ and makes $O(n \log \rk\log\gamma^{-1})$ oracle  calls.
\end{theorem}

By replacing the exponential mechanism with the true maximum, we get a faster, non-private algorithm for \MSD, with the following guarantees.
\begin{corollary}
Instantiated with $\mathrm{MAX}$ instead of $\ExpMech$, \DPNOSG outputs a feasible set $S$ such that $\mathbb{E}[\phi(S)] \ge (\tfrac{1}{2}- \gamma) \phi(\OPT)$ and makes $O(n \log k \log\gamma^{-1})$ oracle  calls.
\end{corollary}

The next theorem states the improved privacy guarantees for decomposable objective functions. Note that under our definition of $\phi'$, if $\phi$ is 1-decomposable, then so is $\phi'$.
\begin{theorem}\label{thm:sample_greedy_decomposable_non_oblivious}
    Suppose $\phi_D:2^V\to \R$ is $1$-decomposable. Then given parameter $\eps_0\in(0,1]$, \Cref{alg:sample_greedy}  is $(\eps,\delta)$-DP for every $\delta>0$ and $\eps=(e^{\eps_0/2}-1)(4+\log(1/\delta))$.
\end{theorem}
\Cref{thm:sample_greedy_guarantee_non_oblivious,thm:sample_greedy_decomposable_non_oblivious} correspond to row \textit{(ii)} in \Cref{tab:results}.
As in the previous section, 
 \Cref{thm:greedy_decomposable} implies that for decomposable objectives, setting $\varepsilon_0 =O( \varepsilon/\log(1/\delta))$ suffices to ensure $(\varepsilon, \delta)$-DP. Substituting this into the utility guarantee of \Cref{thm:sample_greedy_guarantee_non_oblivious} yields the stated additive error in \Cref{tab:results}.

Before proceeding to the proof, we introduce additional notation. Let $\OPT$ be an optimal solution and $S_i$ the set obtained at the end of iteration $i$. Recall that $N_i = V \setminus S_{i-1}$ contains all unselected elements after iteration $i-1$, and $C_i = \OPT \setminus S_{i-1}$ denotes the subset of these elements belonging to $\OPT$. Let $M_i \subseteq N_i$ be a set of size $g(i)$ maximizing $\sum_{v \in M_i} \phi'(v \mid S_{i-1})$. That is, $M_i$ consists of the $g(i)$ elements with the largest marginal contributions during iteration $i$ relative to the auxiliary function $\phi'$. We prove the following.

\begin{lemma}\label{lemma:gain_lowerbound}
    For every iteration $i\in [\rk]$, 
    conditioned on having selected a set $S_{i-1}$ after the first $i-1$ iterations, the following holds.
    $
        \Exp[ \phi'(u_i \mid S_{i-1})] \ge \frac{1-\gamma}{\abs{M_i}}\sum_{v\in M_i} \phi'(v\mid S_{i-1}) - \frac{4\Delta  \log n}{\eps_0}
    $
\end{lemma}

Intuitively, the proof leverages the guarantee of the \ExpMech to show that the expected marginal gain of the selected $u_i$ is at least the average gain across elements in $V_i \cap M_i$ up to a bounded additive error. Since this intersection is non-empty with probability at least $1-\gamma$, and each element in $M_i$ has an equal probability of being sampled into $V_i$, the law of total expectation yields the desired bound. This follows from the definition of $M_i$ as the set of highest contributing elements. The formal proof is provided in \Cref{appendix:sample_greedy}.
\begin{proof}[Proof of \Cref{thm:sample_greedy_guarantee_non_oblivious}]
The complexity bound holds due to the fact that \DPNOSG evaluates $O(\frac{n}{\rk-i} \log \gamma^{-1})$ candidates in iteration $i$, the number of oracle calls across $\rk$ iterations is $O(n \log \rk \log \gamma^{-1})$ using the standard Harmonic sum approximation. 
The full argument is deferred to \Cref{appendix:sample_greedy}.
Privacy follows from the composition theorems (\Cref{thm:composition}) and the $\eps_0$-DP guarantee of the exponential mechanism. We thus focus on the utility guarantee.

Let $f' = \frac{1}{2 - \gamma} f$ and $\phi' = f' + \lambda d$. Define $\alpha = 4\eps_0^{-1}\Delta \log n$. Condition on having selected $S_{i-1}$ prior to iteration $i \in [k]$. By \Cref{lemma:gain_lowerbound}, we have
\begin{align*}
     \Exp[ \phi'(u_i \mid S_{i-1})] &\ge \tfrac{1-\gamma}{|M_i|}\sum_{v\in M_i} \phi'(v\mid S_{i-1}) - \alpha\\
     &\ge \tfrac{1-\gamma}{|C_i|}\sum_{v\in C_i} \phi'(v\mid S_{i-1}) - \alpha \\ 
     &\ge (1-\gamma) \big(\tfrac{f'(\OPT)-f'(S_{i-1})}{\rk} + \tfrac{\lambda(i-1)}{\rk(\rk-1)}d(\OPT)\big) - \alpha.
\end{align*}
The second inequality holds because $|M_i| = k-i+1 \le |C_i|$ , and $M_i$ contains the elements with the largest marginal gains. The third follows from submodularity of $f'$ and \Cref{lemma:dist_lowerbound} for $d$. Taking an expectation over $S_{i-1}$ gives
\begin{align*}
     \Exp[ \phi'(u_i \mid S_{i-1})] \ge  (1-\gamma) \paren{\tfrac{f'(\OPT)-\Exp[f'(S_k)]}{\rk} + \tfrac{\lambda(i-1)}{\rk(\rk-1)}d(\OPT)} - \alpha.
\end{align*}
where we use the monotonicity of $f'$ and the fact that $S_{i-1} \subseteq S_\rk$ to bound  $\Exp[f'(S_\rk)] \ge \Exp[f'(S_{i-1})]$. Summing over $i \in [\rk]$  yields
\[
    \Exp[\phi'(S_\rk)] \ge (1-\gamma)(f'(\OPT)-\Exp[f'(S_\rk)] + \tfrac{\lambda}{2}d(\OPT)) - \rk\alpha.
\]
Rearranging terms and substituting $f'$ gives:
\begin{align*}
    \mathbb{E}[\phi(S_k)] \ge (1-\gamma)\paren{\tfrac{f(\OPT)}{2-\gamma} + \tfrac{\lambda d(\OPT)}{2}} - \rk\alpha \ge (\tfrac{1}{2}-\gamma)\phi(\OPT) - \rk\alpha.
\end{align*}
\end{proof}

 Next, we address the stronger guarantee of~\Cref{thm:sample_greedy_decomposable_non_oblivious}, which applies when the function $\phi$ is decomposable.
We build on the technique of~\cite{gupta2010differentially}, which bounds the privacy loss of the greedy algorithm for submodular maximization by the sum of expected marginal gains of a single individual’s function $f_x$ across all rounds, where the expectation is taken over a distribution induced by the functions of the other individuals. This quantity is bounded with high probability by the sum of realized marginal gains $\sum_{i=1}^{\rk} [f_x(S_i) - f_x(S_{i-1})] = f_x(S_\rk) - f_x(S_0) \leq \Delta$, which telescopes to a constant  independent of $\rk$.
We extend this technique to \Cref{alg:sample_greedy} by providing a unified analysis for the subsampling steps and the \ExpMech. Specifically, we prove that for any deterministic and adaptive choice of $V_i$, the privacy loss can be bounded by a sum of expected marginal gains, which is bounded by a telescoping sum with high probability, analogous to~\cite{gupta2010differentially}. This extension enables our utility bound to outperform a naive composition-based analysis of repeated \ExpMech, and approach existing lower bounds.
We provide a proof sketch, with the full proof given in \Cref{appendix:sample_greedy}.
\begin{proof}[Proof Sketch of \Cref{thm:sample_greedy_decomposable_non_oblivious}]
Let $D$ and $D'$ be two datasets such that $(D\setminus D')\cup (D'\setminus D)=\{x\}$. We first bound the output probability ratio for this add/remove pair, and then apply a reduction to derive the bound for neighboring datasets under the substitution neighboring relation.
Suppose that instead of a set, \Cref{alg:sample_greedy} outputs the sequence of picked elements in their picking order. Let $U= (u_1, \dots, u_\rk)$ be such any sequence.
 By the definition of the exponential mechanism, the probability of selecting element $u$ at iteration $i$ given dataset $D$ and a sampled set of candidates $V_i$ is:
\[
 \Pr[\ExpMech(q^i_D, V_i, D, \eps_0) = u] = \frac{\exp\left(\frac{\eps_0}{2} \beta_D^i(u)\right)}{\sum_{u' \in V_i} \exp\left(\frac{\eps_0}{2} \beta_D^i(u')\right)}.
\]
where we define $\beta_D^i(u) =  \sum_{x\in D} \phi'_x(u\mid \set{u_1,\dots,u_{i-1}})$.
For brevity, denote \Cref{alg:sample_greedy} by $\cG$, and drop irrelevant parameters. 
By the chain rule and the law of total probability over the uniform sampling of $V_i$, the probability ratio $\frac{\Pr[\cG(D)=U]}{\Pr[\cG(D')=U]}$ is bounded by:
\begin{align}
\Bigg( \prod_{i=1}^\rk\tfrac{\exp(\frac{\eps_0}{2}\cdot \beta^i_D(u_i))}{\exp(\frac{\eps_0}{2}\cdot \beta^i_{D'}(u_i))} \Bigg) \Bigg(  \prod_{i=1}^\rk \tfrac{ \sum_{u\in T_i}\exp(\frac{\eps_0}{2}\cdot \beta^i_{D'}(u))}{ \sum_{u\in T_i}\exp(\frac{\eps_0}{2}\cdot \beta^i_D(u))}\Bigg). \label{eq:total_ratio}
\end{align}
 where $T_i$ is a  worst-case realization of $V_i$ containing $u_i$. We consider two cases.

\textbf{Case 1: $D = D' \cup \{x\}$.} By the $1$-decomposability of $\phi'$, the first factor in \eqref{eq:total_ratio} is bounded by $\exp(\frac{\eps_0}{2}\sum_i (\beta^i_D(u_i) - \beta^i_{D'}(u_i))) \le \exp(\eps_0/2)$. The second factor is at most $1$ since $\beta^i_x \ge 0$.

\textbf{Case 2: $D' = D \cup \{x\}$.} The first factor in \eqref{eq:total_ratio} is at most $1$ . The second factor is a product of expectations $\prod_{i=1}^\rk \E{u\sim P_i}[\exp(\frac{\eps_0}{2} \beta^i_x (u))]$, where $P_i(v) \propto \exp(\varepsilon_0 \beta^i_D(v))$ for $v \in T_i$. Using the concentration bounds from \cite{gupta2010differentially,chaturvedi2021differentially}, the product of expectations is bounded by $\exp((e^{\eps_0/2}-1)(3+\log(1/\delta)))$ with probability $1-\delta$.

Combining these add/remove cases and applying a standard reduction establishes an $((e^{\eps_0/2}-1)(4+\log(1/\delta)), \delta)$-DP guarantee for neighboring datasets under the substitution relation.
\end{proof}

So far, we have analyzed \DPNOSG, an instantiation of \Cref{alg:sample_greedy} that makes  $O_{\gamma}(n \log \rk)$ oracle calls. In the following section, we introduce a linear-time algorithm where the number of oracle calls is independent of $\rk$.

\subsection{Oblivious Sample Greedy}\label{para:obl_sample_greedy}
We present \DPOSG, an efficient algorithm for DP \MSD that makes only $O_\gamma(n)$ oracle calls, albeit with a weaker approximation factor. Our experiments in \Cref{sec: exp} demonstrate that \DPOSG achieves high utility under tight privacy constraints. Furthermore, it is significantly more scalable than \DPG, making it better suited for interactive settings.
\DPOSG uses the utility function $g(i) = \min\set{\rk ,\,  n-i+1}$ for all $i\in[\rk]$.
A key component in the analysis of \DPNOSG is \Cref{lemma:dist_lowerbound}, which lower bounds the total distance between the current solution and the remaining unselected elements of $\OPT$ as a function of $d(\OPT)$. While this yields tight guarantees for \DPG and \DPNOSG, applying the same technique with the current choice of $g$ leads only to a $1/6-\gamma$ approximation.
We therefore take a different approach and instantiate \Cref{alg:sample_greedy} with $\phi'=\phi$, i.e., in an oblivious manner. Instead of relying on \Cref{lemma:dist_lowerbound}, we bound the total distance to the remaining elements of $\OPT$ using the marginal gain $d(\OPT\cup S_i)-d(S_i)$ (\Cref{lemma:dist_marginal_gain}). This yields a new progress inequality for the oblivious variant, leading to a tighter approximation factor.

One can verify that the privacy guarantees stated in \Cref{thm:sample_greedy_decomposable_non_oblivious,thm:sample_greedy_guarantee_non_oblivious}
hold for \Cref{alg:sample_greedy} with the current choice of $g$ and $\phi'$. 

\begin{restatable}{theorem}{ObliviousSampleGreedy}\label{thm:sample_greedy_guarantee_oblivious}
Suppose $\phi$ has sensitivity $\Delta$. Then,  \DPOSG (\Cref{alg:sample_greedy}) outputs a feasible set $S$ such that 
\[\Exp[\phi(S)] \ge (1-(\tfrac{2}{e})^{(1-\gamma)(1-1/\rk)}) \cdot \phi(\OPT) - O(\frac{\rk \Delta \log n}{\eps_0})\] and makes $O(n\log\gamma^{-1})$ oracle  calls.
\end{restatable}
\ifpaper
Lemma C.4 in~\cite{full_version}
\else
\Cref{lemma:technical_prod}
\fi
shows that the approximation factor is greater than $1-2/e-\gamma-1/\rk$, and we use this linearized form for simplicity.
\Cref{thm:sample_greedy_guarantee_oblivious}, together with the privacy guarantees in \Cref{thm:sample_greedy_decomposable_non_oblivious,thm:sample_greedy_guarantee_non_oblivious} correspond to row \textit{(iii)} in \Cref{tab:results}.
By replacing the exponential mechanism with the true maximum, we get the first linear-time, non-private algorithm with a constant factor approximation.
\begin{corollary}
    Instantiated with $\mathrm{MAX}$ instead of $\ExpMech$, \DPOSG outputs a feasible set $S$ such that
    $       \Exp[\phi(S)] \ge (1-(\tfrac{2}{e})^{(1-\gamma)(1-1/\rk)}) \cdot  \phi(\OPT),
    $
    and makes $O(n\log\gamma^{-1})$ oracle  queries.
\end{corollary}
The proof of \Cref{thm:sample_greedy_guarantee_oblivious} relies on the following structural progress lemma, whose proof is deferred to~\Cref{appendix:sample_greedy}.
\begin{restatable}{lemma}{DistMarginalGain}\label{lemma:dist_marginal_gain}
    For every  iteration $i\ge 2$,  
    \[
         d(\OPT\cup S_{i-1}) - d(S_{i-1}) \le \Big(1+\frac{\rk -1}{i-1}\Big)\cdot d(C_{i},S_{i-1}).
    \]
\end{restatable}

\ifpaper
\begin{proof}[Proof Sketch of \Cref{thm:sample_greedy_guarantee_oblivious}]
\ifpaper
\else
The full complexity analysis is provided in \Cref{appendix:sample_greedy}. 
\fi
Intuitively, since \DPOSG evaluates $O(\frac{n}{\rk} \log \gamma^{-1})$ candidates in each iteration, the number of oracle calls across $\rk$ iterations is $O(n \log \gamma^{-1})$.  We focus here on the utility guarantee.

Let  $\alpha = 4\eps_0^{-1}\Delta \log n$. Consider any  iteration $2\le i \le \rk$ and condition on having selected set $S_{i-1}$ after iteration $i-1$. By submodularity and monotonicity,
$
        \sum_{u\in C_i} f(u\mid S_{i-1}) \ge f(\OPT\cup S_{i-1}) - f(S_{i-1}).
$        
By \Cref{lemma:dist_marginal_gain},
\[
     \sum_{v\in C_i} d(v\mid S_{i-1})  = d(C_i,S_{i-1}) \ge  \tfrac{d(\OPT\cup S_{i-1}) -d(S_{i-1})}{1+\frac{\rk -1}{i-1}}.
\]
Let \mbox{$M_i \subseteq N_i$} be a set of size $g(i)=\min\set{\rk ,\abs{N_i}}$ that maximizes 
\mbox{$
\sum_{v \in M_i} \phi(v \mid S_{i-1}).
$}
By definition of $M_i$, we have
\begin{align*}
     &\Exp[ \phi(u_i \mid S_{i-1})] \ge \tfrac{1-\gamma}{g(i)}\sum_{v\in M_i} \phi(v\mid S_{i-1}) - \alpha\\
     &\ge  \tfrac{1-\gamma}{\rk}\sum_{v\in C_i} \phi(v\mid S_{i-1}) - \alpha    
      \ge (1-\gamma)\Bigg(\tfrac{\phi(\OPT)-\phi(S_{i-1})}{\rk \Big(1+\tfrac{\rk -1}{i-1}\Big)}\Bigg) - \alpha 
\end{align*}
The  first inequality follows by an argument analogous to that in the proof of \Cref{lemma:gain_lowerbound}.
The second inequality follows from the definition of $M_i$, together with the fact that $|C_i| \leq g(i) = |M_i|\leq \rk$. Removing the conditioning on $S_{i-1}$ and taking an expectation  over all its possible realizations yields
\begin{align}
     \Exp[ \phi(u_i \mid S_{i-1})] \ge (1-\gamma)\Bigg(\tfrac{\phi(\OPT)-\phi(S_{i-1})}{\rk \Big(1+\frac{\rk -1}{i-1}\Big)}\Bigg) - \alpha  
\end{align}
Rearranging gives
    \[\phi(\OPT) - \Exp[\phi(S_i)] \le \Big(1-\tfrac{1-\gamma}{\rk (1+\frac{\rk -1}{i-1})}\Big)  \paren{ \phi(\OPT) - \phi(S_{i-1}) } + \alpha.\]

By recursively applying the last inequality after removing the conditioning on $S_{i-i}$, taking the expectation over all its realizations, we get
\begin{align*}
    \phi(\OPT) - \Exp[\phi(S_\rk)] 
    & \le  \prod_{i=2}^\rk \Big(1-\tfrac{1-\gamma}{\rk (1+\frac{\rk -1}{i-1})}\Big)  \phi(\OPT) +\rk \alpha \\
    &\le \paren{\tfrac{2}{e}}^{(1-\gamma)(1-1/\rk)} \phi(\OPT) +\rk \alpha
\end{align*}
where the second inequality is by Lemma C.4 in~\cite{full_version}. The theorem now follows by rearranging.
\end{proof}
\else

\begin{proof}[Proof of \Cref{thm:sample_greedy_guarantee_oblivious}]
The proof of the query complexity is provided in \Cref{lemma:sample_greedy_complexity}. We thus focus on the utility guarantee. Let $\alpha = 4\Delta \log (n)/\eps_0$. Consider any iteration $2 \le i \le \rk$ and condition on the set $S_{i-1}$ selected after iteration $i-1$. By the submodularity and monotonicity of $f$, we have:
\[
    \sum_{u\in C_i} f(u \mid S_{i-1}) \ge f(\OPT \cup S_{i-1}) - f(S_{i-1}) \ge f(\OPT) - f(S_{i-1}).
\]
Applying \Cref{lemma:dist_marginal_gain}, we find that:
\[
     \sum_{v\in C_i} d(v \mid S_{i-1}) = d(C_i, S_{i-1}) \ge \frac{d(\OPT \cup S_{i-1}) - d(S_{i-1})}{1 + \frac{\rk - 1}{i - 1}}.
\]
Let $M_i \subseteq N_i$ be a set of size $g(i) = \min\{\rk, |N_i|\}$ that maximizes $\sum_{v \in M_i} \phi(v \mid S_{i-1})$. Note that $M_i$ comprises the $\rk$ elements with the largest marginal contributions in iteration $i$, or is equal to $N_i$ if $|N_i| \le \rk$. We have:
\begin{align*}
     &\Exp[ \phi(u_i \mid S_{i-1})] \ge \tfrac{1-\gamma}{g(i)}\sum_{v\in M_i} \phi(v\mid S_{i-1}) - \alpha \ge \tfrac{1-\gamma}{\rk}\sum_{v\in C_i} \phi(v\mid S_{i-1}) - \alpha \\
     & \ge (1-\gamma) \Bigg( \frac{\phi(\OPT) - \phi(S_{i-1})}{\rk (1 + \frac{\rk - 1}{i - 1})} \Bigg) - \alpha.
\end{align*}
The first inequality follows from the guarantees of the exponential mechanism and the sampling procedure, analogous to the proof of \Cref{lemma:gain_lowerbound}. For the second inequality, recall that $C_i \subseteq N_i$, and since $C_i \subseteq \OPT$, it holds that $|C_i| \le \rk$. Thus, the inequality follows from the definition of $M_i$, together with the fact that $|C_i| \le g(i) \le \rk$.

Removing the conditioning on $S_{i-1}$ and taking the expectation over all possible realizations yields:
\[
     \Exp[ \phi(u_i)] \ge (1-\gamma) \Bigg( \frac{\phi(\OPT) - \Exp[\phi(S_{i-1})]}{\rk (1 + \frac{\rk - 1}{i - 1})} \Bigg) - \alpha.
\]
Rearranging terms, we obtain:
\[
    \phi(\OPT) - \Exp[\phi(S_i)] \le \left(1 - \frac{1-\gamma}{\rk (1 + \frac{\rk - 1}{i - 1})}\right) \left( \phi(\OPT) - \Exp[\phi(S_{i-1})] \right) + \alpha.
\]
By recursively applying this inequality, we get:
\begin{align*}
    \phi(\OPT) - \Exp[\phi(S_{\rk})] &\le \prod_{i=2}^{\rk} \left(1 - \frac{1-\gamma}{\rk (1 + \frac{\rk - 1}{i - 1})}\right) \phi(\OPT) + \rk \alpha \\
    &\le \left(\frac{2}{e}\right)^{(1-\gamma)(1 - 1/\rk)} \phi(\OPT) + \rk \alpha,
\end{align*}
where the last inequality follows from \Cref{lemma:technical_prod}. The theorem follows by rearranging the final expression.
\end{proof}

\fi

\section{Local Search Algorithm for Matroid Constraints}
\label{sec:matroid_local_search}
In this section, we present our algorithm for matroid constraints. \citet{borodin2017max} proposed a non-private local search achieving a $(1/2-\gamma)$-approximation in $O(n^2+\gamma^{-1}n\rk^2\log \rk)$ oracle calls\footnote{The algorithm in \cite{borodin2017max} attains a $1/2$-approximation but is not polynomial-time. The authors' proposal to perform only swaps yielding at least a $\gamma$-improvement at each iteration, rather than any improvement, leads to the stated guarantees.}  starting from an initial base containing the highest-scoring feasible pair of elements. However, privatizing this approach is non-trivial: additive DP noise prevents both verifying local optimality and guaranteed selection of strictly improving swaps.

Our algorithm, \DPSLS\ (\Cref{alg:sample_local_search}), executes a fixed number of local search iterations, using the exponential mechanism to privately select an approximately best available swap from a subsampled candidate set. To avoid forced suboptimal swaps when no improving swap exists, we introduce a dummy swap that allows the current solution to remain unchanged. Finally, we apply the exponential mechanism to privately select and output the highest-scoring solution observed throughout the execution.

Our analysis exploits the structure of the underlying metric space to establish sufficient expected progress and achieve a near-optimal approximation ratio, while avoiding higher-order error terms introduced by continuous relaxations~\cite{rafiey2020fast}. We show that whenever the expected objective value in a given iteration is low, the subsequent iteration increases it significantly. In particular, the first iteration already guarantees a sufficiently large expected objective value, enabling a multiplicative progress argument that yields a $(1/2 - \gamma)$-approximation within $O(\gamma^{-1}\rk\log \rk)$ iterations.

By basing our approach on local search, we depart from predominantly greedy methods used in DP submodular maximization. This approach yields a combinatorial proof with a tighter additive error term, while simultaneously eliminating the $O(n^2)$ factor present in the complexity of~\cite{borodin2017max}. We further reduce complexity by subsampling candidate swaps in each iteration, ultimately requiring $O(\gamma^{-1}n\rk\log \rk)$ oracle calls. This provides a substantial improvement over prior work while maintaining the same near-optimal expected approximation factor.

\begin{theorem}\label{thm:sample_local_search_guarantee}
   Suppose $\phi_D:2^V\to\R$ has sensitivity $\Delta$, and let $\cM$ be a matroid of rank $\rk$. 
For $T' = O(\gamma^{-1} \rk\log \rk)$, \DPSLS (\Cref{alg:sample_local_search}) with parameters $\eps_0 > 0$ and $\gamma \in (0,1)$ 
is $T'\eps_0$-DP, and $(\eps,\delta)$-DP for every $\delta > 0$, 
with $\eps=\sqrt{2T'\log(1/\delta)}\eps_0+T'\eps_0(e^{\eps_0}-1)$.
    Moreover, it outputs a feasible set $S$ such that
    \[
\mathbb{E}[\phi_D(S)] \ge  \paren{\tfrac{1}{2} - \gamma}\cdot \phi_D(\OPT)
- O\left( \tfrac{\Delta }{\eps_0}  \left( \rk \log  n  +\log\tfrac{1}{\gamma}) \right) \right), \]
    and makes $O(\gamma^{-1}n \rk\log \rk )$ oracle  calls.
\end{theorem}

\Cref{thm:sample_local_search_guarantee} corresponds to row \textit{(iv)} in \Cref{tab:results}. The stated additive error in \Cref{tab:results} is obtained using advanced composition (\Cref{thm:composition}), which intuitively implies that setting $\varepsilon_0 = O(\varepsilon / \sqrt{\gamma^{-1}\rk \log \rk\log\delta^{-1}})$ ensures $(\varepsilon, \delta)$-DP. Note that we now have $O(\gamma ^{-1}\rk \log \rk )$ compositions.
By replacing the exponential mechanism with the true maximum, we get a faster, non-private algorithm for the \MSD  problem with matroid constraints, for which the expected approximation guarantee is the same as that of the deterministic local search algorithm of \cite{borodin2017max}.
\begin{corollary}
   Instantiated with $\mathrm{MAX}$ instead of $\ExpMech$, \DPSLS  outputs a feasible set $S$ such that
    $
        \Exp[\phi(S)] \ge \paren{\tfrac{1}{2} - \gamma}  \cdot \phi(\OPT),
    $
    and makes $O(\gamma^{-1}n\rk \log \rk  )$   oracle  calls.
\end{corollary}

\begin{remark}\label{ls_remark} In this section, we consider the broader class of $\Delta$-sensitive functions. 
The technique of \citet{gupta2010differentially} for decomposable functions relies crucially on the monotonic growth of the greedy solution sequence. 
Extending this approach to local-search algorithms, where elements are both added and removed, appears nontrivial and remains an interesting direction for future work.
\end{remark}

\begin{algorithm}
\caption{DP Sample Local Search (\DPSLS)}
\label{alg:sample_local_search}
\begin{algorithmic}[1]
\Require Dataset $D$,
ground set $V$,
submodular $f_D$,
pseudometric $d_D$, matroid $\cM$ of rank $\rk$, privacy parameter $\eps_0$, utility parameter $\gamma\in(0,1]$,
diversity parameter $\lambda$

\State Let $\phi_D(\cdot) = f_D(\cdot)+\lambda d_D(\cdot)$
\State Let $S_0$ be an arbitrary base of $\cM$.
\State Set $T \gets \lceil \tfrac{2\rk\log (8\rk)}{\gamma(1-1/e)}\rceil + 1$ 
\For{$i=1,\dots,T$}
        \State Sample a subset $V_i\subseteq V$ of size $\lceil\tfrac{n}{\rk}\rceil$ uniformly at random \plabel{localsearch_sampling}
        \State Let $W_i \gets\set{(u,v)\in S_{i-1}\times (V_i\setminus S_{i-1}) \mid S_{i-1} -u + v\in \cI}\cup\set{(w,w)}$ for an arbitrary element $w\in S_{i-1}$.
        \State For all $(u,v)\in W_i$,  define
        $q_D^i(u,v)= \phi_D(S_{i-1} -u + v)$.
\State Compute $(u_i, v_i)  \gets  \ExpMech(D, q_D^i, W_i,\eps_0)$
\State Let $S_i \gets S_{i-1} -u_i + v_i$
\EndFor
\State For all $i=1,\dots,T$, define $s_D(i) = \phi_D(S_i)$.
\State Compute $i^* \gets \ExpMech(D, s_D, [T],\eps_0)$ \plabel{localsearch_EM}
\State \Return $S_{i^*}$
\end{algorithmic}
\end{algorithm}

For the remainder of this section we prove  \Cref{thm:sample_local_search_guarantee}. For simplicity, we assume throughout that the $\cM$ has rank $\rk \ge 3$, and address the case $\rk = 2$ 
\ifpaper
in~\cite{full_version}.
\else
in~\Cref{appendix:rank2}.
\fi

Given a set $S$ and elements $u,v$, we use $S+u$, $S-u$ and $S-u+v$  as shorthands for $S\cup \set{u}$, $S\setminus\set{u}$ and $(S\setminus\set{u})\cup\set{v}$ respectively.  Let $\OPT$ be an optimal solution.  Since $\phi$ is monotone, we may assume, without loss of generality, that $\OPT$ is a base of $\cM$. For $i\in[T]$, let $S_{i-1}$ denote the base obtained after iteration $i-1$. 

\begin{lemma}\label{lemma:sum_swaps_lowerbound}
    Suppose $\cM$ has rank $\rk>2$. There exists a bijection
$h:S_{i-1}\to \OPT$ such that  $h(u)=u$ for every $u\in S_{i-1}\cap \OPT$, and $\OPT-h(u)+u\in\cI$ for every $u\in S_{i-1}$. Moreover,
    \[
        \sum_{u\in S_{i-1}}\phi(S_{i-1}-u+h(u)) \ge \phi(\OPT)+(\rk-2)\phi(S_{i-1}).
    \]
\end{lemma}
Using this lemma, we establish the following result, characterizing the algorithm’s progress.
\begin{lemma}\label{lemma:sample_progress}
For every iteration $i\in[T]$, we have
    \begin{align*}
    \Ex{\phi(S_{i}) -\phi(S_{i-1}) } \ge   \tfrac{1-1/e}{\rk}\paren{\phi(\OPT)-2\Ex{\phi(S_{i-1})}} - \tfrac{4\Delta \log n }{\eps_0}
\end{align*}
\end{lemma}
Intuitively, \Cref{lemma:sum_swaps_lowerbound} implies that after obtaining a base $S_{i-1}$, a uniformly random swap from the set $M_i = \set{(u,h(u)) : u \in S_{i-1}}$ yields an expected gain of $\tfrac{1}{\rk}(\phi(\OPT) - 2\phi(S_{i-1}))$. To prove \Cref{lemma:sample_progress}, we show that the set of swaps considered at iteration $i$ intersects $M_i$ with probability at least $1-1/e$. Conditioned on this event, the swap selected by the exponential mechanism achieves, up to an additive error, the maximal gain available, which in turn is at least as large as the average gain of a swap in $M_i$. Combining these observations establishes the lemma. The full argument, which also handles potential negative gains, is presented in
~\Cref{appendix:localsearch}.
\begin{proof}[Proof Sketch of \Cref{thm:sample_local_search_guarantee}]
The privacy guarantee follows from the $\eps_0$-DP of the exponential mechanism and composition theorems (\Cref{thm:composition}) over $O(\gamma^{-1} \rk \log \rk)$ iterations. 
For complexity, each iteration requires $O(|S_{i-1}| \cdot |V_i|) = O(n)$ oracle calls due to the subsampling step. No additional calls are required in Line~\ref{localsearch_EM} as all values are computed during the preceding step. In total, the algorithm makes $O(\gamma^{-1} n \rk \log \rk)$ oracle calls.
For utility, we first obtain $\Ex{\phi(S_1)} = \Omega(\frac{\phi(\OPT)}{\rk})$ using \Cref{lemma:sample_progress}, assuming $\phi(\OPT)$ is large enough that the bound in the theorem is non-trivial. Conditioning on the selection of $S_{i-1}$, \Cref{lemma:sample_progress} implies:
\begin{align}\Ex{\phi(S_i)} \ge \paren{1+\frac{\gamma(1-1/e)}{\rk}}\Ex{\phi(S_{i-1})} \label{eq:expectation_progress1}
\end{align}
for every $i$ for which the bound $\Ex{\phi(S_{i-1})} < \phi(\OPT)/(2+\gamma) - \rk\alpha$ holds. Thus, there must exist some $i$ for which the bound fails; otherwise, recursive application of  \eqref{eq:expectation_progress1} yields $\Ex{\phi(S_{T})} > \phi(\OPT)$ in contradiction. The proof is completed by applying the law of total expectation over the sampling of $S_1, \dots, S_{T}$ and concluding that the selected $S_{i^*}$ has expected quality at least $\max_{i \in [T]}\Ex{{\phi(S_i)}}$ up to a bounded additive error. This maximum is at least $\phi(\OPT)/2$ up to the stated additive error.\end{proof}
\section{Experiments}\label{sec: exp}
In this section, we evaluate the utility and efficiency of our proposed algorithms across two data summarization applications: Uber pickup location summarization and Amazon product summarization. Based on these applications, we characterize the trade-offs between privacy, utility, and complexity.

\paratitle{Summary of Findings}
For cardinality constraints, the proposed DP algorithms achieve competitive utility compared to the non-private baseline while drastically improving efficiency. \DPG and \DPNOSG stay within a marginal $1\%$ of the non-private baseline even at a strict $\eps = 0.1$ budget. By shifting the query complexity to logarithmic or zero dependence on $k$, \DPNOSG and \DPOSG achieve speedups of $14\times$ and $50\times$ at $k=100$. For matroid constraints, \DPSLS maintains competitive utility, staying within $1\%$ of the non-private baseline at $\eps=0.1$. While \DPSLS reduces the total number of oracle calls, its fixed iteration count leads to an execution time overhead of up to $2.2\times$ compared to the non-private baseline; however, despite scaling more slowly with $k$, it exhibits the improved scaling with $n$ predicted by our analysis.

\subsection{Experimental Settings}\label{sec:settings}
We next present our settings for the experiments. 
All algorithms were implemented\footnote{
The implementation can be found at \url{https://github.com/ronzadi/Differentially-Private-Max-Sum-Diversification}.
} in Python 3.9.19 using the Pandas and NumPy libraries.
All experiments were run on an Intel Xeon CPU-based server with 24 cores
and 96 GB of RAM.  

\paratitle{Default parameters}
 Unless mentioned otherwise, the following parameters are used.
 We set $\eps = 0.1$ and $\delta=|D|^{-1.5}$ where $D$ is the sensitive dataset. For cardinality constraints, our per-iteration privacy parameter $\eps_0$ can be set with either basic composition, advanced composition (e.g., \Cref{thm:greedy_guarantee}) or the analysis for decomposable objective (e.g., \Cref{thm:greedy_decomposable}). We follow \cite{mitrovic2017differentially,chaturvedi2021differentially} and pick the initialization that achieves $(\eps,\delta)$-DP with the smallest noise scale. For the partition matroid constraint, we pick the best out of basic and advanced composition.
We fix $\gamma = 0.1$ for the utility parameter, and by default use  $\lambda=0.5$. Results are averaged over 10 runs.

\paratitle{Baselines}
As \MSD has not been previously explored in a DP setting, we compare our proposed methods against non-private, random, and DP submodular maximization baselines.

\begin{itemize}[leftmargin=*]
    \item  \textbf{\Greedy} (Cardinality): The standard non-private greedy algorithm for cardinality constraints \cite{borodin2017max}.
    \item \textbf{\LS} (Matroid): The non-private local search algorithm for general matroid constraints \cite{borodin2017max} discussed in \Cref{sec:matroid_local_search}.
    \item \textbf{\Random} (Cardinality/Matroid): Selects a random feasible subset. Since the output distribution of the exponential mechanism converges to  uniform as $\eps \to 0$, this baseline serves as a lower bound on utility for private algorithms. Following prior work~\cite{mitrovic2017differentially,chaturvedi2021differentially,chaturvedi2023streaming}, we incorporate \Random to contextualize the performance of our algorithms.
        \item \textbf{\DPGRel} (Cardinality / Matroid): The DP greedy algorithm of~\cite{mitrovic2017differentially} for DP submodular maximization. It greedily selects feasible elements using the \ExpMech, applied only to the submodular  component of our objective (i.e., $\lambda=0$). This baseline isolates the effect of removing the diversity term by optimizing relevance alone.
    \item \textbf{\DPMix} (Cardinality): A hybrid baseline combining the random baseline with \DPGRel: $k/2$ items are selected  uniformly at random, and the remaining $k/2$ are selected using \DPGRel.
\end{itemize}

\subsection{Datasets, Objective Functions, and Constraints}
We evaluate our approach on two real-world data summarization tasks: Uber location selection and Amazon product selection.

\paratitle{Uber Location Selection} 
The first application involves selecting a representative summary of Uber pickup locations in Manhattan. This setup adapts the location selection experiments from \cite{mitrovic2017differentially, chaturvedi2021differentially,mitrovic2018data}. The goal is to select a set of locations from a public candidate set, for instance, to serve as waiting spots for idle drivers.

\begin{itemize}[leftmargin=*, label={}]
\item \textbf{Dataset.}
We consider  a sensitive dataset $D$ of $m=600,000$ Uber pickup locations in Manhattan from 2014 \cite{uberDataset}. Each record consists of longitude and latitude coordinates. For the public ground set, we follow the established methodology and consider a grid of $n$ points over Manhattan, where by default $n=1000$.

\item \textbf{Objective function.}
For relevance, we adopt the Uber ``convenience'' score from~\citet{mitrovic2018data}, defined as follows. For a passenger location
    $a=(x_a,y_a)$ and a grid point $b=(x_b,y_b)$ define $c(a,b) = 2 - \frac{2}{1 + e^{-200 \norm{a-b}_1}}$, 
    where $\norm{a-b}_1= |x_a-x_b| + |y_a-y_b|$ is the Manhattan distance between the two points.
The relevance of a set $S$ is defined as $f_D(S) = \frac{1}{m} \sum_{a \in D} \min_{b\in S}c(a,b)$
with $f_D(\emptyset) = 0$. This function is monotone, submodular, and \mbox{$1/m$-decomposable} (See
\ifpaper
Appendix~F in~\cite{full_version}
\else
\Cref{appendix:sens_uber}
\fi). However, as observed in \cite{mualem2022using}, it does not inherently promote geographic diversity. For example, suppose congestion slows down traffic in and out of a certain area. If all selected waiting locations are concentrated there, drivers may struggle to reach passengers efficiently. To address this, we add a diversity term defined by the normalized sum of distances between the selected points: $d(S)=\sum_{b,b'\in S}\frac{\norm{b-b'}_1}{M}$, where $M$ is an upper bound on the distance within the area such that $\norm{b-b'}_1\in [0,1]$

\item \textbf{Constraints.} We evaluate this task under a cardinality constraint~$k$.
\end{itemize}

\paratitle{Amazon Product Selection}
The second application follows the scenario in \Cref{example:motivating} and focuses on selecting a representative summary of highly-purchased Amazon products to maximize user reach.
\begin{itemize}[leftmargin=*, label={}]
    \item \textbf{Dataset.} We utilize the \emph{Health and Household} subset of the Amazon Reviews 2023 dataset \cite{amazon}, which contains product metadata (e.g., categories, price) and user review data. We treat the metadata as the public ground set and the review/purchase history as the private dataset. Specifically, we select  $3000$ products from the \emph{Health Care} category as our ground set $V$ and discretize their prices into four bins. The dataset $D$ contains the purchase records of $1{,}317{,}193$ users.
    \item \textbf{Objective function.}
We consider the objective function from our running example (\Cref{example:decomposable}), where the goal is to select a subset $S \subseteq V$ maximizing user reach, with diversity given by the Jaccard distance over category sets.

\item \textbf{Constraints.} We consider
the intersection of a uniform matroid of rank $k$ with a partition matroid that limits selections to $\lceil k/4 \rceil$ products per price bin. We refer to the latter setting as Amazon-Partition.
We also conducted experiments under a cardinality constraint. The scalability trends were similar to those observed for Uber, which is expected since the two settings differ primarily in the cost of evaluating the objective function rather than in the algorithmic behavior. We likewise observed highly similar utility trends, and therefore defer these results to
\ifpaper
Appendix~G in~\cite{full_version}
\else
\Cref{appendix:experiments}
\fi
\end{itemize}

\begin{figure*}
    \centering
    \begin{minipage}[t]{0.24\textwidth}
        \centering
        \includegraphics[width=\linewidth]{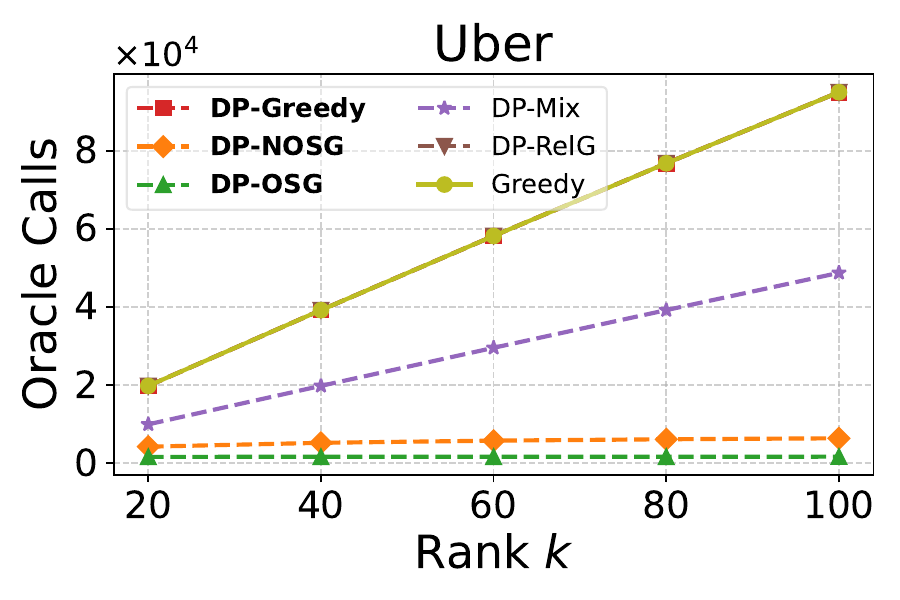} \\ \vspace{1pt}
        \includegraphics[width=\linewidth]{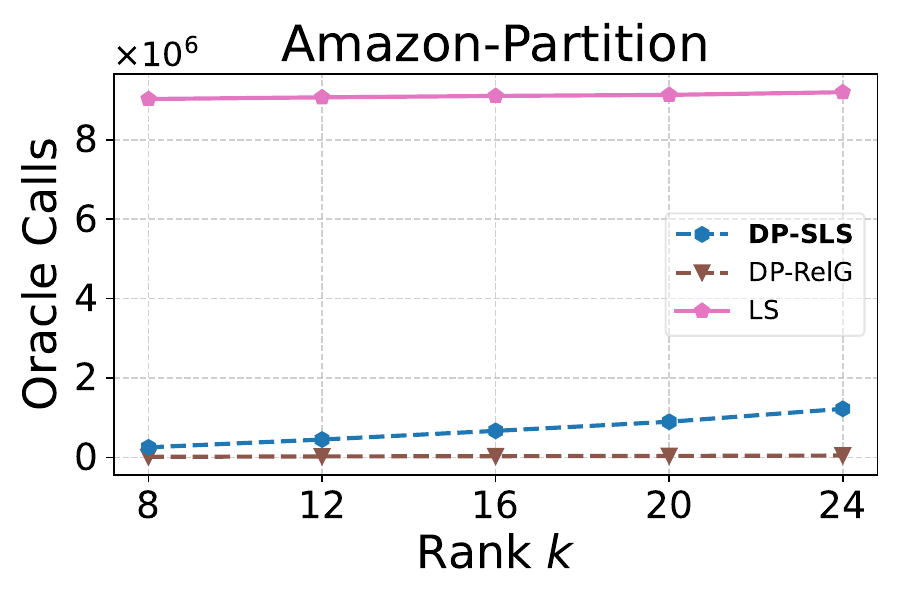} \\ \vspace{1pt}
        \caption{Number of oracle  calls as the rank $k$ varies.}
        \label{fig:queries_for_k}
    \end{minipage}
    \hfill
    \begin{minipage}[t]{0.24\textwidth}
        \centering
        \includegraphics[width=\linewidth]{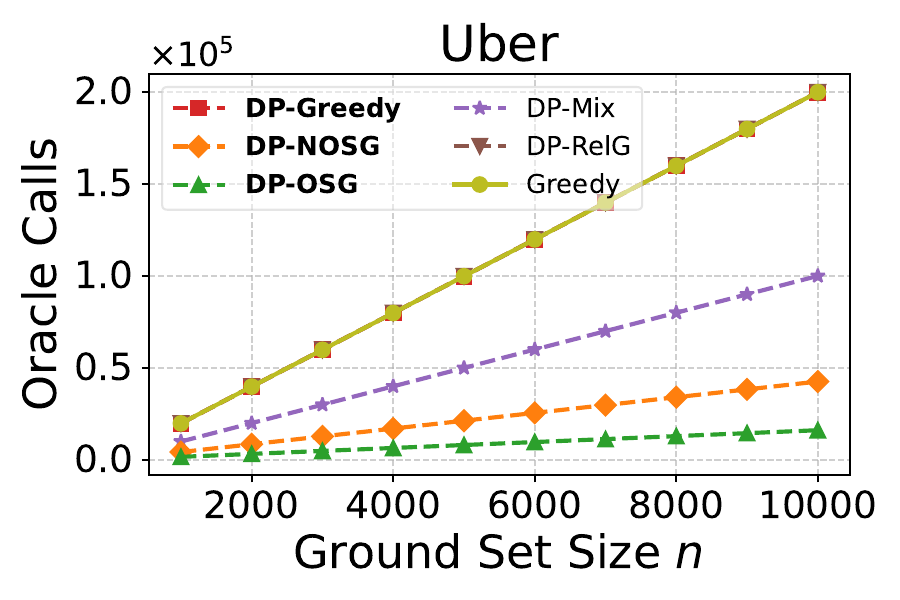} \\ \vspace{1pt}
        \includegraphics[width=\linewidth]{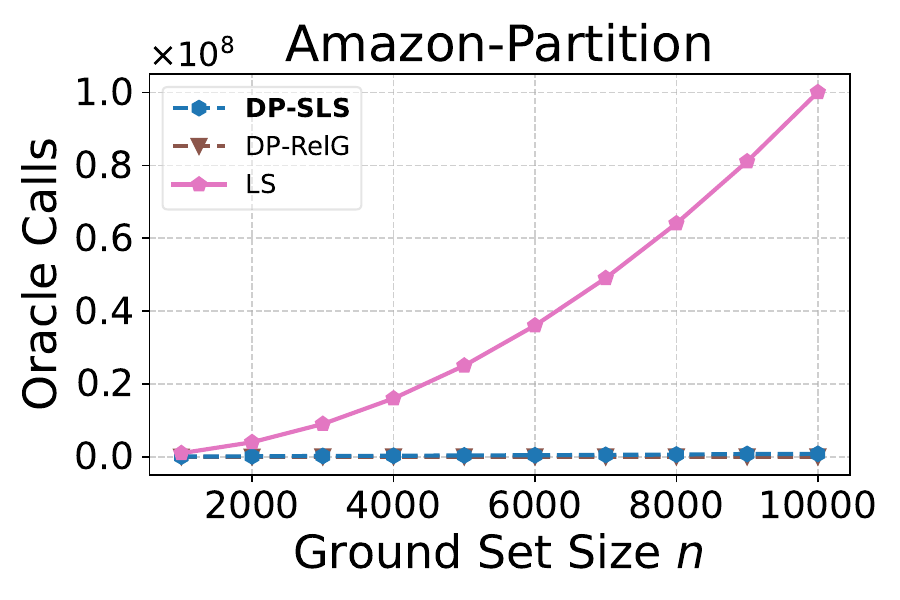} \\ \vspace{1pt}
        \caption{Number of oracle  calls as the ground set size $n$ varies.}
        \label{fig:queries_for_n}
    \end{minipage}
    \hfill
    \begin{minipage}[t]{0.24\textwidth}
        \centering
        \includegraphics[width=\linewidth]{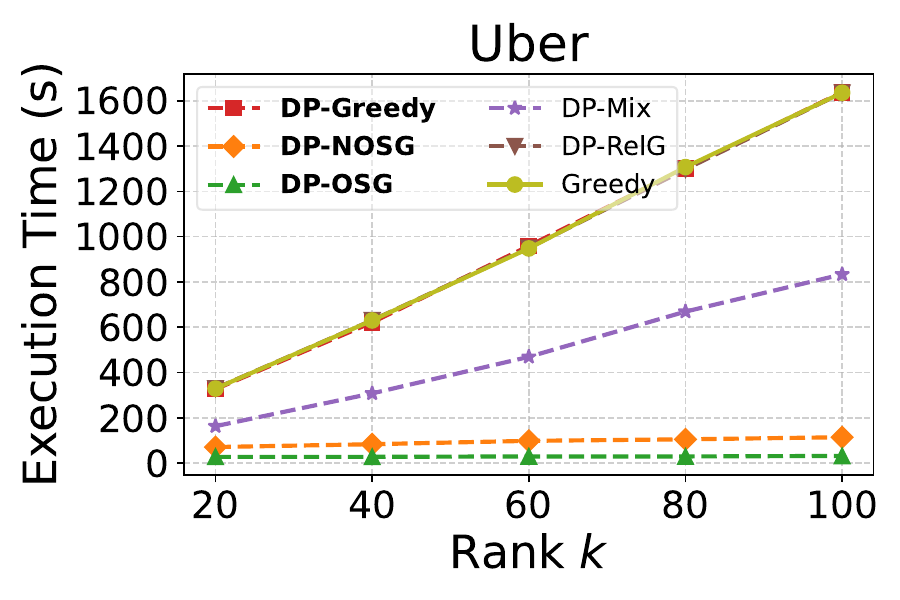} \\ \vspace{1pt}
        \includegraphics[width=\linewidth]{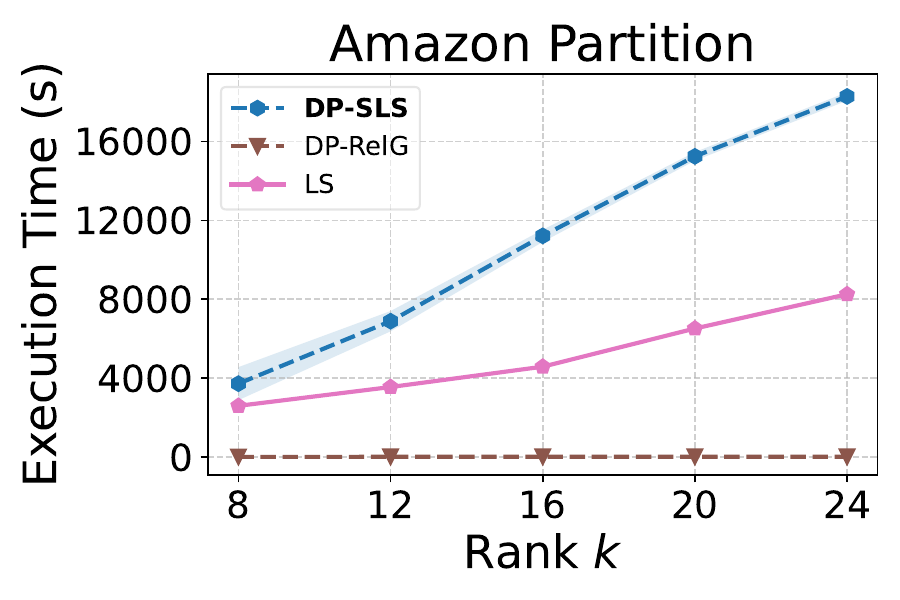} \\ \vspace{1pt}
         \caption{Execution time (seconds) as the rank $k$ varies.}
        \label{fig:time_for_k}
    \end{minipage}
    \hfill
    \begin{minipage}[t]{0.24\textwidth}
        \centering
        \includegraphics[width=\linewidth]{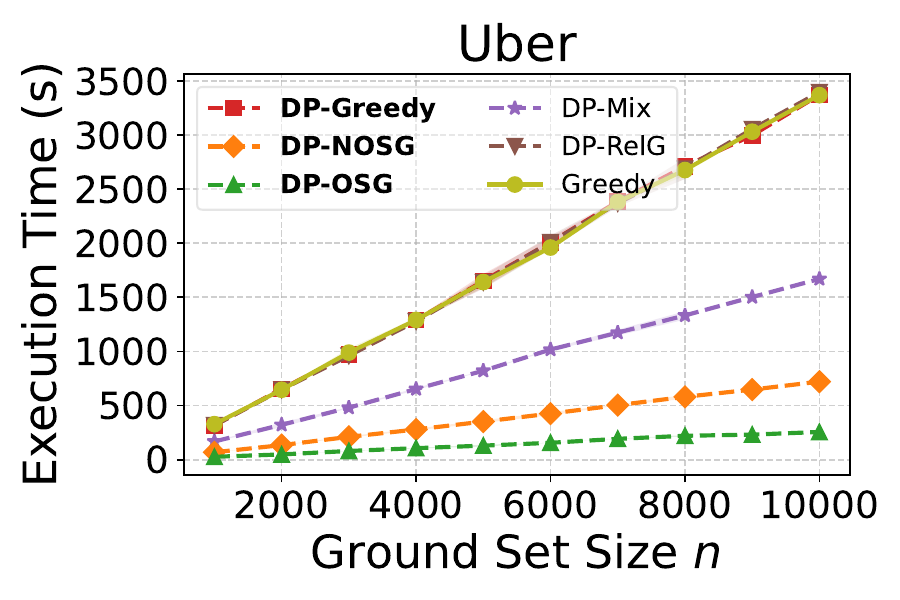} \\ \vspace{1pt}
        \includegraphics[width=\linewidth]{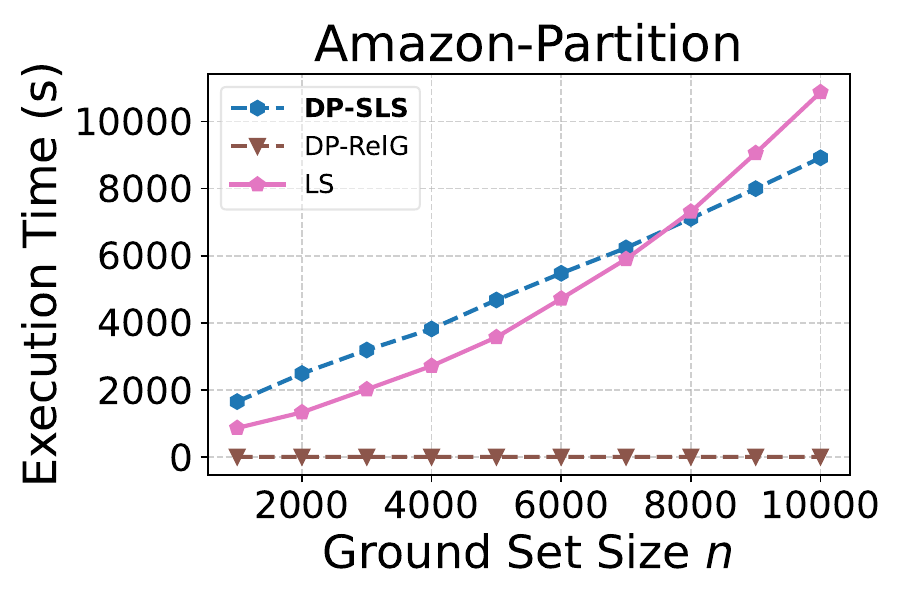} \\ \vspace{1pt}
        \caption{Execution time (seconds) as the ground set size $n$ varies.}
        \label{fig:time_for_n}
    \end{minipage}
\end{figure*}

\subsection{Utility Analysis}
\label{subsec: acc}
We examine how the selected subset size $k$ and the privacy parameter $\eps$ impact the \MSD objective value $(1-\lambda)f_D(S)+\frac{2\lambda}{k(k-1)}d(S)$.
By default, we set $k=20$ for the Uber experiment, $k=6$ for the Amazon-Partition experiment.

\paratitle{Impact of $k$}
We evaluate the impact of the rank $k$ on the objective value (\Cref{fig:val_for_k}). Overall, our algorithms perform competitively with the non-private baselines and consistently outperform the other DP baselines.
For the Uber dataset, our three algorithms remain within $3.4\%$ of the non-private \Greedy baseline on average, with the gap increasing in $k$ due to accumulated DP noise. In contrast, \DPMix is on average $17\%$ below \Greedy, \DPGRel is $20\%$ below, and \Random is $26\%$ below. The trends for Amazon-Cardinality are similar and the corresponding plot is omitted.
For Amazon-Partition, the utility of \DPSLS remains nearly identical to the non-private \LS baseline across the examined range, and is sometimes even slightly better (within $0.1\%$ on average). In contrast, the baselines \DPGRel and \Random perform substantially worse, remaining approximately $25\%$ and $30\%$ below \LS, respectively. We also note that the objective value need not increase monotonically with $k$. This is because the normalization factor of the diversity increases with $k$, and stricter partition constraints for smaller $k$ may restrict the algorithm to naturally more diverse solutions.

\paratitle{Impact of $\eps$}
We evaluate the impact of the privacy parameter $\eps$ on the resulting objective value (\Cref{fig:val_for_eps}). As expected, a larger privacy budget improves the utility of all DP algorithms.
For the Uber dataset, all our algorithms significantly outperform the DP baselines and achieve utility comparable to the non-private \Greedy baseline, reaching within $1\%$ even at $\eps=0.1$. In contrast, \DPMix, \DPGRel, and \Random remain substantially lower, at $19\%$, $27\%$, and $34\%$ below, respectively.
For Amazon-Cardinality, the trends are similar and the corresponding plot is omitted. For Amazon-Partition, we similarly find that at $\eps=0.1$, \DPSLS reaches utility within $1\%$ of the non-private \LS baseline and significantly outperforms \DPGRel and \Random, which remain $38\%$ and $48\%$ below \LS, respectively. \DPGRel slightly benefits from lower privacy budgets, as the increased randomness induces more diverse selections, resulting in performance closer to \Random in this regime.

\paratitle{Impact of $\lambda$}
We evaluate the robustness of our algorithms with respect to $\lambda$, which controls the trade-off between relevance and diversity. The results, shown in \Cref{fig:lambda}, demonstrate that the utility of our algorithms relative to the corresponding non-private baselines remains consistently high across all examined values $\lambda \in [0,0.8]$, indicating that the observed performance is largely insensitive to the choice of $\lambda$.
For the Uber dataset, the utility gap from the non-private \Greedy baseline remains below $1\%$ on average for \DPG and \DPNOSG, and below $3\%$ for \DPOSG. In contrast, the other DP baselines perform substantially worse, with \Random reaching up to $38\%$ lower utility on average.
For Amazon-Partition, the utility gap of \DPSLS from the non-private \LS baseline remains below a marginal $0.1\%$ across all examined values of $\lambda$. Similar trends are observed for Amazon-Cardinality.

As expected, the performance of relevance-only baselines deteriorates as $\lambda$ increases as diversity becomes more dominant in the objective. When $\lambda=0$, corresponding to relevance-only optimization, our algorithms achieve essentially the same utility as \DPGRel. This shows that our approach recovers the performance of prior work in this special case while naturally extending to more general relevance-diversity objectives.

\begin{figure}
    \centering
    \begin{minipage}[t]{0.23\textwidth}
        \centering
        \includegraphics[width=\linewidth]{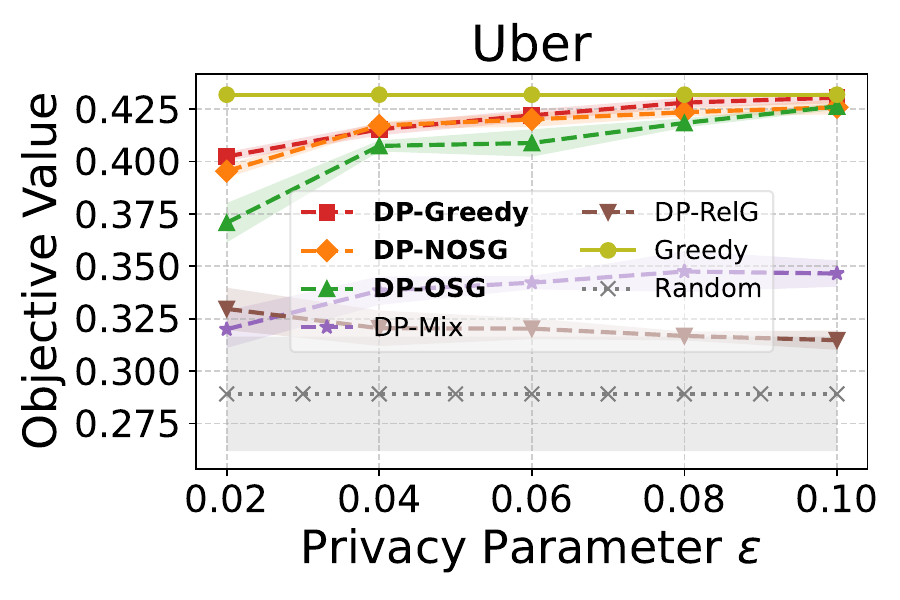} \\  
        \includegraphics[width=\linewidth]{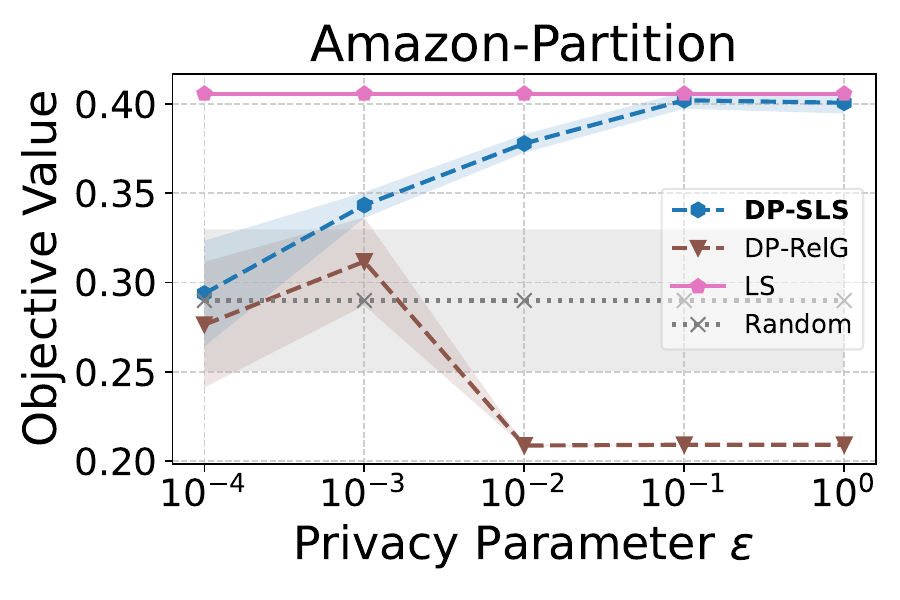} \\ 
        \caption{Objective value as the privacy parameter $\eps$ varies.}
        \label{fig:val_for_eps}
    \end{minipage}
    \hfill
    \begin{minipage}[t]{0.23\textwidth}
        \centering
        \includegraphics[width=\linewidth]{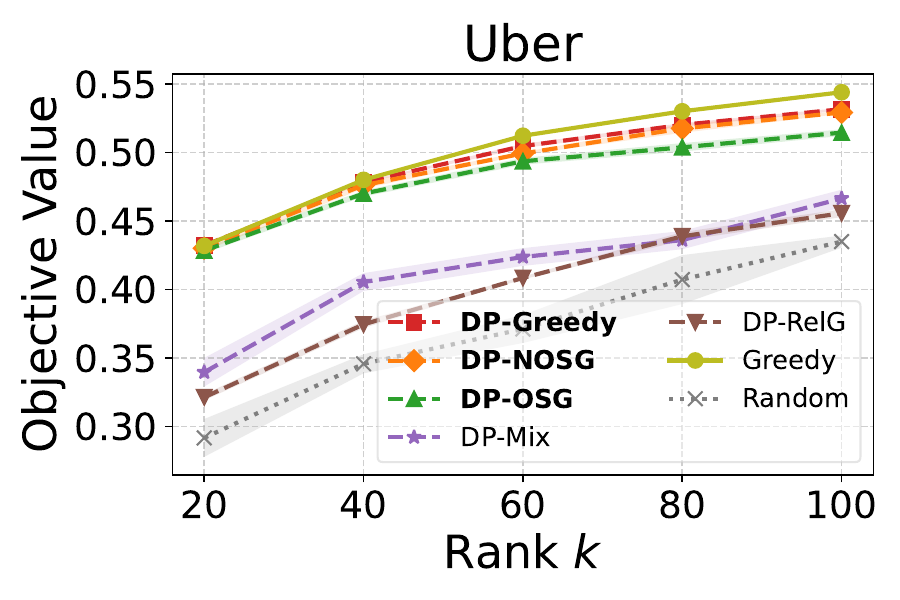} \\                
        \includegraphics[width=\linewidth]{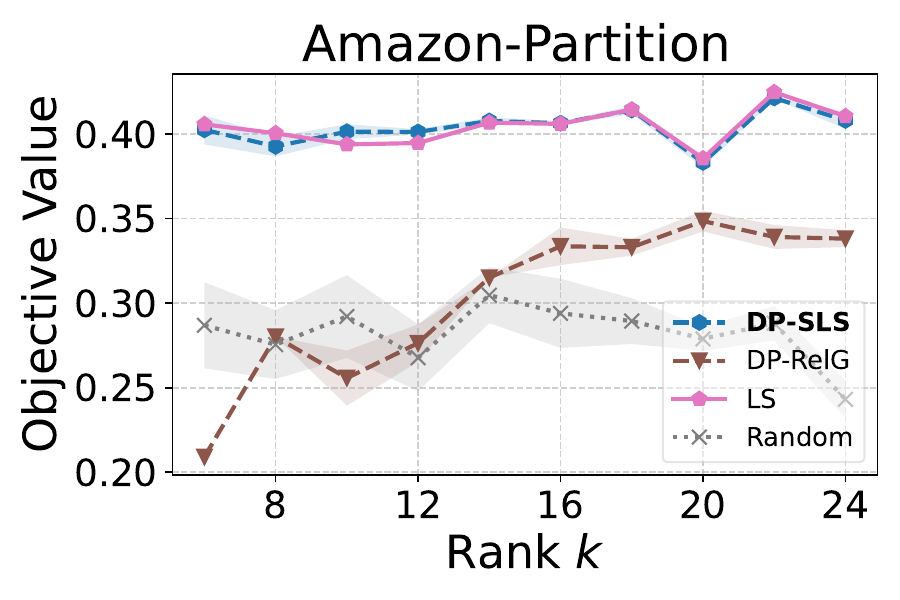} \\ 

         \caption{Objective value as the rank (maximal size of a feasible solution) $k$ varies.}
        \label{fig:val_for_k}
    \end{minipage}
\end{figure}

\begin{figure}[t]
    \centering
    \begin{minipage}[t]{0.23\textwidth}
        \centering
        \includegraphics[width=\linewidth]{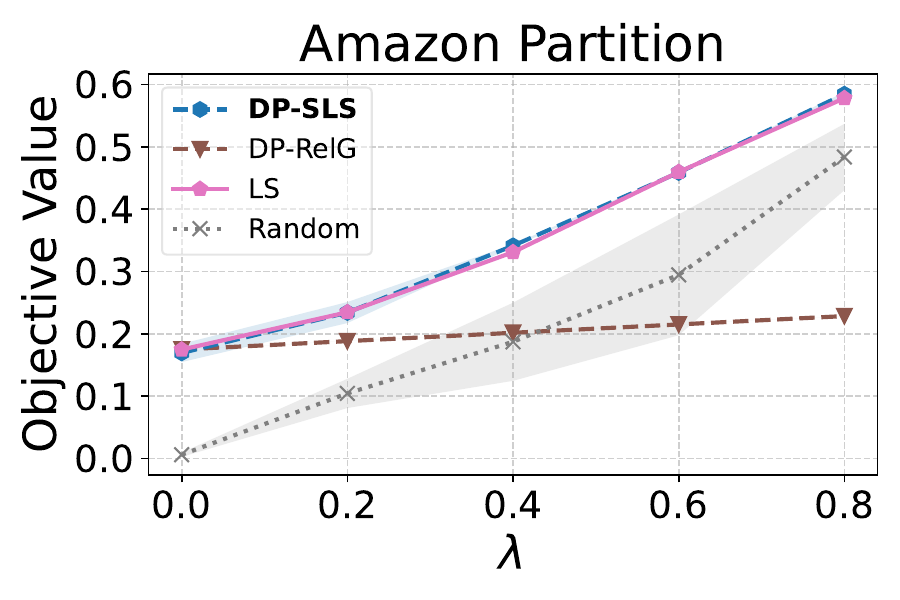}
    \end{minipage}%
    \begin{minipage}[t]{0.23\textwidth}
        \centering
        \includegraphics[width=\linewidth]{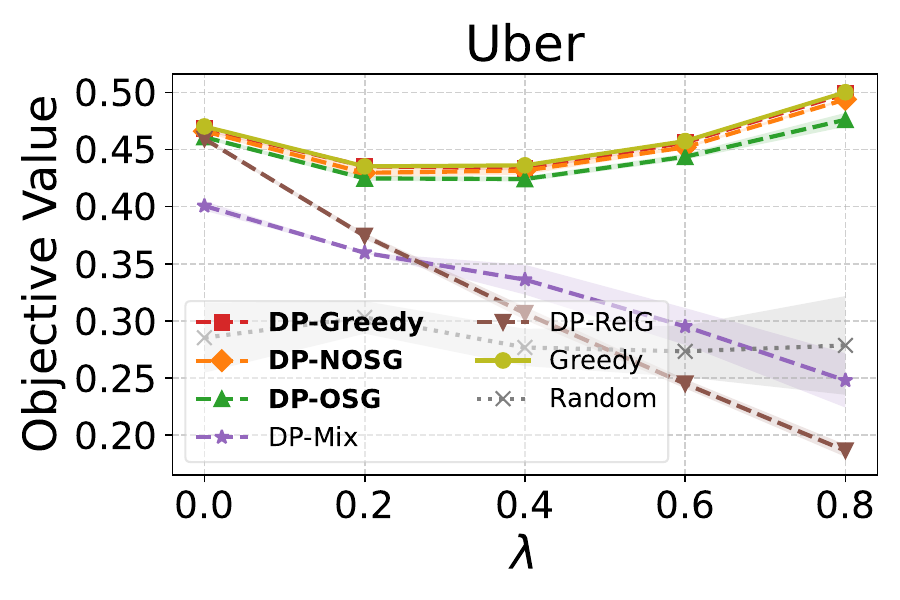}
    \end{minipage}
    \caption{Objective value as the parameter $\lambda$ varies.}
    \label{fig:lambda}
\end{figure}

\subsection{Performance Analysis}
\label{subsec:perf}
We evaluate the computational efficiency of our algorithms by measuring both query complexity and execution time (\Cref{fig:queries_for_k,fig:time_for_k}). In the submodular maximization literature, the number of oracle  calls is considered the most robust metric for performance evaluation, as it is invariant to implementation and problem-specific details~\cite{li2022submodular}. Nevertheless, we also report execution time to provide additional practical insights. Note that we do not attempt to optimize the oracle evaluation time.

\paratitle{Number of Oracle Calls}
\Cref{fig:queries_for_k} reports the number of oracle calls as $k$ varies. For cardinality constraints, \DPG matches the non-private \Greedy baseline, scaling linearly in $k$, while \DPNOSG exhibits a logarithmic dependence and significantly reduces query complexity (e.g., $15\times$ fewer calls than \Greedy at $k=100$ on Uber). \DPOSG is even more efficient, with a query count essentially independent of $k$, achieving up to a $57\times$ reduction at $k=100$ while maintaining comparable utility. 
For Amazon-Partition, the query count of the non-private \LS baseline is dominated by its initialization phase, which requires an exhaustive $\Theta(n^2)$ search over feasible pairs and is independent of $k$; in practice, the subsequent local search terminates after fewer iterations than the worst-case bound. In contrast, \DPSLS avoids this exhaustive initialization and maintains a lower query count.

\Cref{fig:queries_for_n} shows the dependence on $n$. On the Uber dataset, all algorithms exhibit linear scaling in $n$, while our \DPNOSG and \DPOSG achieve a smaller slope due to their improved dependence on $k$. For Amazon-Partition, \LS exhibits quadratic scaling in $n$ due to its initialization phase, which involves a brute-force search over pairs of elements, whereas our \DPSLS scales linearly.

\paratitle{Execution Time}
\Cref{fig:time_for_k} illustrates execution time as $k$ varies. For cardinality constraints, the trends align with our query complexity analysis. \DPNOSG exhibits a logarithmic dependence on $k$, outperforming both \Greedy and all DP baselines; on Uber ($k=100$), it is $14\times$ faster than \Greedy. \DPOSG is even more efficient, achieving a $50\times$ speedup while preserving comparable utility.

For Amazon-Partition, \DPSLS is slower than the non-private \LS baseline despite its fewer oracle calls. This discrepancy arises because the \LS bottleneck is dominated by pair evaluations, which are computationally cheaper in this application than evaluating the larger sets required for potential swaps. Moreover, the non-private baseline often converges to a local optimum in substantially fewer iterations than the fixed number required by the \DPSLS theoretical analysis.
Overall, \DPSLS is approximately $2\times$ slower than \LS on average. While both algorithms may be slow for interactive settings, \DPSLS uniquely provides formal privacy and utility guarantees. In contrast, \DPGRel is significantly faster (minutes rather than hours, and is thus omitted from the plot), but achieves substantially lower utility and lacks a constant-factor approximation guarantee.

\Cref{fig:time_for_n} depicts scalability as $n$ varies. For cardinality constraints, \DPOSG and \DPNOSG are the most scalable DP methods, and their execution-time trends match the query-complexity trends. For Amazon-Partition, \DPSLS exhibits the improved scalability in $n$ predicted by our bottleneck analysis; for sufficiently large ground sets ($n \ge 8000$), it becomes faster than the \LS baseline.

\section{Conclusion and Limitations}
\label{sec:conclusion}
We studied the max-sum diversification problem, a widely adopted formulation introduced by \citet{borodin2017max}, under differential privacy. We designed DP algorithms for both cardinality and general matroid constraints. These algorithms achieve utility guarantees comparable to those of their non-private counterparts while also offering improved complexity, making them attractive even when privacy is not required. Experiments on real-world datasets demonstrate that our approach attains high utility and substantially improves practical efficiency for cardinality constraints.

Our work has several limitations that suggest interesting directions for future research. First, as discussed in \Cref{ls_remark}, our current guarantees for matroid constraints apply to arbitrary $\Delta$-sensitive functions and do not exploit decomposability to obtain stronger utility bounds. Achieving such improvements appears to be non-trivial. Second, our fastest $O_{\gamma}(n)$ oracle-call algorithm provides a constant-factor approximation but does not attain the tight $1/2$ ratio. It remains open whether a near-optimal $(1/2-\gamma)$-approximation can be achieved for cardinality constraints with $O_{\gamma}(n)$ calls. Third, while \Cref{alg:sample_local_search} has strictly better worst-case query complexity than the non-private baseline of \citet{borodin2017max}, it can be slower in practice in some regimes due to realizing its worst-case complexity in every execution. Developing practical accelerations while preserving its high utility is an interesting direction. Finally, extending our framework to additional diversity models remains an important avenue for future work.

\begin{acks}
This research was supported by the Israel Science Foundation (ISF) under grant 2707/22 of the Breakthrough Research Grant (BRG) Program.
\end{acks}
\newpage

\balance
\bibliographystyle{ACM-Reference-Format}
\bibliography{bibl}

\ifpaper
\else
\onecolumn
\appendix
\section{Submodular Functions and Matroids}\label{appendix:matroids}
Submodular set functions have numerous applications across areas such as machine learning, databases, economics, and network analysis. Their usefulness stems from the property of diminishing returns, which naturally arises in many practical situations, and informally states that the incremental benefit of adding an element to a smaller set is larger than adding it to a larger set. This property makes submodular functions  well-suited to model the notion of relevance. 
\begin{definition}\cite{nemhauser1978analysis}
 A set function $f:2^V\to\R$ is submodular
 if $f(u\mid S) \ge f(u\mid T)$ for every two sets
$S\subseteq T\subseteq V$ and element $u\in V\setminus T$.
\end{definition}
 A set function $f:2^V\to\R$ is \emph{non negative} if $f(S)\ge 0$ for every set $S\subseteq V$, and is \emph{monotone} if $f(S)\le f(T)$ for every two sets $S\subseteq T\subseteq V$. 
We consider only non-negative, monotone submodular functions, as is the case in the non-private setting~\cite{borodin2017max}.

The problem of maximizing a submodular function subject to a \emph{matroid
independence} constraint is one of earliest and most studied problems in submodular maximization~\cite{schrijver2003combinatorial}. In the context of max-sum diversification as well, matroids allow a substantial generalization of the types of constraints that can be modeled compared to simple cardinality \cite{borodin2017max}. For example, cardinality constraint is a special case called a \emph{uniform matroid}. In a \emph{partition matroid},  the ground set $V$ is partitioned into disjoint subsets $P_i$, and the independent sets are 
$\set{S \in 2^V : |S \cap P_i| \le r_i}$, 
i.e., sets that satisfy a separate cardinality bound for each block.
As another example, a \emph{transversal matroid} models representativeness. Let $A_1, \dots, A_m$ a collection of categories, i.e, subsets of the ground set $V$ (that may overlap). A subset $S \subseteq V$ is independent if there exists an injective mapping $h: S \to [m]$ such that $u \in A_{h(u)}$ for all $u \in S$, that is, no two items represent to the same category. Additional background an examples can be found in~\cite{schrijver2003combinatorial}. We now give the formal definition.
\begin{definition}\cite{whitney1992abstract}
    A matroid is a pair $(V,\cI)$, where $\cI\subseteq 2^V$ is a collection of subsets called \emph{independent sets}, with the following properties.
    \begin{enumerate}[label=(\textit{\roman*})]
        \item if $A\subseteq B$ and $B\in \cI$ then $A\in \cI$ (Hereditary),
        \item If $A,B\in \cI $ and $|A|<|B|$ then $\exists u\in B\setminus A$ such that  $A\cup\set{u}\in \cI$ (Augmentation).
    \end{enumerate}
    
\end{definition}
A matroid is a combinatorial abstraction of the notion of independence from linear algebra. Following the terminology
used for linear spaces, independent sets that are inclusion-wise maximal are called \emph{bases}. An immediate
corollary of the augmentation property is that all bases have the same number of elements, and this number is referred to as the \emph{rank} of the matroid. 
As is standard in the literature, we assume access to the set functions and the matroid via \emph{oracles}. Specifically, the \emph{value  oracle} for $f$ (respectively, $d$) returns the value $f(S)$ (respectively, $d(S)$) given a set $S$. 
The \emph{independence oracle} associated with a matroid $(V, \cI)$ takes as input a set $S \subseteq V$ and answers whether $S \in \cI$.
The number of oracle  calls performed by an algorithm is commonly used as a proxy for
its time complexity, as the exact complexity may be application specific, or depend on implementation details such as  the data structures used to maintain sets. In our algorithms, the number of queries to the value and independence oracles are the same up to constant factors. We therefore refer to the total number of queries to both as the number of \emph{oracle  calls}.
\section{Omitted Proofs from Section \texorpdfstring{\protect\lowercase{\ref{sec:greedy}}}{SECTION \ref{sec:greedy}}}\label[appendix]{appendix:greedy}
This appendix contains the proofs of \Cref{lemma:dist_lowerbound,thm:greedy_guarantee}

We first recall the notation. Let $\OPT$ be an optimal solution, i.e.,  a subset of size at most $\rk$ that maximizes $\phi$. Since $\phi$ is monotone, we may assume that $\abs{\OPT}=\rk$. Let $S_i$ be the solution at the end of step $i$. Define $A_i = S_{i-1}\cap \OPT$, $B_i=S_{i-1}\setminus A_i$ and $C_i=\OPT\setminus A_i$. 

\begin{lemma}[\cite{ravi1994heuristic}]\label{lemma:ravi_distance}
Suppose $d$ is a pseudometric, and let $X$ and $Y$ be two disjoint subsets of~$V$. Then, $|Y|\cdot d(X)\le (|X|-1)\cdot d(Y,X)$.
\end{lemma}
\begin{proof}[Proof of \Cref{lemma:ravi_distance}]
    If either set is empty, or if $|X|=1$ then both sides equal zero and the inequality trivially holds. Thus, we assume that both sets are not empty and that $|X|\ge 2$. For distinct $x,x'\in X$ and for $y\in Y$, we have by the triangle inequality $d(x,x') \le d(x,y)+d(x',y)$. Summing over all distinct pairs $x,x'\in X$, we obtain $d(X)\le (|X|-1)d(y,X)$, because each distance $d(y,x)$ appears in the sum exactly $|X|-1$ times. Summing over all $y\in Y$, we get $|Y|d(X)\le (|X|-1)d(Y,X)$.

\end{proof}

We proceed with the proof of  \Cref{lemma:dist_lowerbound}, which we restate here for convenience. 
\distlowerbound*
\begin{proof}[Proof of \Cref{lemma:dist_lowerbound}]
    We first prove the claim for the special case of $C_i = \set{o}$ for some $o\in V\setminus S_{i-1}$, which is not covered by the proof in \cite{borodin2017max}.
    Since $|S_{i-1}|$ increases by $1$ in each round, it follows that $i=\rk$. Hence, it remains to prove that $d(o,S_{i-1}) \ge d(\OPT)/\rk$.
    Note that our assumption implies that $\OPT = S_i \cup \set{o}$. By \Cref{lemma:ravi_distance}, we have 
    \[
        d(S_{i-1}) \le(\abs{S_{i-1}}-1) d(o, S_{i-1}) \le (\rk-1) d(o, S_{i-1}).
    \]
Therefore, 
\begin{align*}
    d(\OPT) = d(S_{i-1}\cup\set{o}) = d(S_{i-1}) + d(o,S_{i-1}) \le \rk\cdot d(o, S_{i-1}).
\end{align*}
For the case $\abs{C_i}>1$, we follow the proof of \citet{borodin2017max}.
By \Cref{lemma:ravi_distance}, the following inequalities hold
\begin{align}
    &(\abs{C_i}-1)d(B_i,C_i) \ge \abs{B_i}d(C_i) \label{eq:borodin1}\\
    &(\abs{C_i}-1)d(A_i,C_i) \ge \abs{A_i}d(C_i) \label{eq:borodin2}\\
    &(\abs{A_i}-1)d(A_i,C_i) \ge \abs{C_i}d(A_i) \label{eq:borodin3}\\
    &d(A_i,C_i) + d(A_i)+d(C_i) = d(\OPT) \label{eq:borodin4}
\end{align}
where \eqref{eq:borodin4} holds since $A_i$ and $C_i$ are disjoint sets whose union equals $\OPT$. Then, we multiply equations \eqref{eq:borodin1},\eqref{eq:borodin2},\eqref{eq:borodin3},\eqref{eq:borodin4} by the following non-negative numbers respectively
\[
    \frac{1}{\abs{C_i}-1},\;\frac{\abs{C_i}-\abs{B_i}}{\rk(\abs{C_i}-1)},\; \frac{i-1}{\rk(\rk-1)},\;\frac{(i-1)\abs{C_i}}{\rk(\rk-1)}.
\]
Summing the multiplied equations, we get
\[
    d(A_i,C_i)+d(B_i,C_i) - \frac{(i-1)\abs{C_i}(\rk-\abs{C_i})}{\rk(\rk-1)(\abs{C_i}-1)}d(C_i) \ge \frac{(i-1)\abs{C_i}}{\rk(\rk-1)}d(\OPT).
\]
Since $\abs{C_i}\le \rk$ and $ d(A_i,C_i)+d(B_i,C_i)=d(C_i,S_{i-1})$,
we have
\[
d(C_i,S_{i-1}) \ge \frac{(i-1)\abs{C_i}}{\rk(\rk-1)}d(\OPT),
\]
and the proof is complete.    
\end{proof}
We are now ready to complete the proof of \Cref{thm:greedy_guarantee}.
\begin{proof}[Proof of \Cref{thm:greedy_guarantee}]
For the query complexity, observe that the algorithm performs $\rk$ iterations, each requiring $O(n)$ oracle queries. The privacy guarantee follows directly from the composition theorems (\Cref{thm:composition}), together with the $\eps_0$-DP property of the exponential mechanism. 

Thus, we proceed with the utility guarantee. Let $f'(S) = f(S)/2$, so that $\phi'(S) = f'(S) + \lambda d(S)$. Consider any iteration $i \in [\rk]$, and let us condition on the set $S_{i-1}$ selected after the first $i-1$ iterations. We have:
\begin{align*}
      & \sum_{u \in C_i} f'(u \mid S_{i-1}) \ge f'(C_i \cup S_{i-1}) - f'(S_{i-1}) \\
      &= f'(\OPT \cup S_{i-1}) - f'(S_{i-1}) \ge f'(\OPT) - f'(S_{i-1}),
\end{align*}

where the first inequality holds by the submodularity of $f'$, and the second by monotonicity. Moreover, by \Cref{lemma:dist_lowerbound}:
\[
    \sum_{u \in C_i} d(u \mid S_{i-1}) = d(C_i, S_{i-1}) \ge \frac{(i-1)|C_i|}{\rk(\rk-1)} d(\OPT).
\]
Letting $\alpha = 4\eps_0^{-1}\Delta \log n$, we obtain:
\begin{align*}
    \Exp[\phi'(u_i \mid S_{i-1})] &\ge \max_{u \in V \setminus S_{i-1}} \phi'(u \mid S_{i-1}) - \alpha\\
    &\ge \frac{1}{|C_i|} \left( \sum_{u \in C_i} \phi'(u \mid S_{i-1}) \right) - \alpha \\
    &\ge \frac{f'(\OPT) - f'(S_{i-1})}{\rk} + \frac{\lambda(i-1)}{\rk(\rk-1)} d(\OPT) - \alpha.
\end{align*}
The first inequality follows from the guarantees of the exponential mechanism, along with the fact that the quality functions $q_D^i$ have sensitivity at most $2\Delta$. The second inequality holds since $C_i \subseteq V \setminus S_{i-1}$.

Removing the conditioning on $S_{i-1}$ and taking the expectation over all its possible realizations, the last inequality yields:
\[
    \Exp[\phi'(u_i \mid S_{i-1})] \ge \frac{f'(\OPT) - \Exp[f'(S_{\rk})]}{\rk} + \frac{\lambda(i-1)}{\rk(\rk-1)} d(\OPT) - \alpha,
\]
where we have used the fact that $\Exp[f'(S_{i-1})] \le \Exp[f'(S_{\rk})]$ due to the monotonicity of $f'$. Summing over all $i \in [\rk]$, we get:
\[
    \Exp[\phi'(S_{\rk})] \ge f'(\OPT) - \Exp[f'(S_{\rk})] + \frac{\lambda}{2} d(\OPT) - \rk\alpha.
\]
Plugging in $\phi' = f/2 + \lambda d$ and $f' = f/2$, we obtain:
\[
    \frac{1}{2} \Exp[f(S_{\rk})] + \Exp[\lambda d(S_{\rk})] \ge \frac{f(\OPT) - \Exp[f(S_{\rk})] + \lambda d(\OPT)}{2} - \rk\alpha.
\]
Rearranging terms gives $\Exp[\phi(S_{\rk})] \ge \frac{1}{2} \phi(\OPT) - \rk\alpha$, which completes the proof.
\end{proof}

\section{Omitted Proofs from Section \texorpdfstring{\protect\lowercase{\ref{sec:sample_greedy}}}{\ref{sec:sample_greedy}}}
\label[appendix]{appendix:sample_greedy}
This appendix contains the missing proofs for the results presented in \Cref{sec:sample_greedy}. 

\begin{proof}[Proof of \Cref{lemma:gain_lowerbound}]
Let $\alpha = 4\eps_0^{-1}\Delta \log n$ and $i \in [k]$. Condition on a realization of $S_{i-1}$ and $V_i$. By \cref{thm:exponential_mech} and the $2\Delta$ sensitivity bound for $q^i_D$, we have
\begin{align*}
    \Ex{\phi'(u_i \mid S_{i-1})} &\ge \max_{v\in V_i}\phi'(v \mid S_{i-1}) - \alpha \ge \max_{v\in V_i\cap M_i}\phi'(v \mid S_{i-1}) - \alpha \\ 
    &\ge \sum_{v\in V_i\cap M_i}\tfrac{\phi'(v \mid S_{i-1})}{|V_i \cap M_i|} -\alpha.
\end{align*}
Note that second inequality holds even if $V_i \cap M_i = \emptyset$ because $\phi'$ is monotone and marginal gains are nonnegative. Removing the conditioning on $V_i$ and taking the expectation over all its realizations gives
\begin{equation}
\label{eq1}
\begin{split}
    \Exp[ \phi'(u_i \mid S_{i-1})] &\ge \E{V_i}\Big[\sum_{v\in V_i\cap M_i}\tfrac{\phi'(v \mid S_{i-1})}{|V_i \cap M_i|}\Big]- \alpha \\
    &\ge \tfrac{1-\gamma}{|M_i|}\sum_{v\in M_i}\phi'(v\mid S_{i-1}) - \alpha.
\end{split}
\end{equation}
To prove the last inequality, let $G$ be the event $V_i \cap M_i \neq \emptyset$. If $\log(1/\gamma) \ge g(i)$, then $\Pr[G]=1$ holds. Otherwise, the sampling in Line~\ref{cardinality_sampling} ensures
\[
    \Pr[G] = \Pr[V_i\cap M_i \neq \emptyset] \ge 1- \paren{1-\tfrac{g(i)}{|N_i|}}^{\frac{|N_i|\log(1/\gamma)}{g(i)}} \ge 1-\gamma.
\]
Let $p_v$ be the probability that $v \in M_i$ appears in $V_i$ conditioned on $G$. By the law of total expectation:
\begin{align*}
    \E{V_i}\Big[\sum_{v\in V_i\cap M_i}\tfrac{\phi'(v \mid S_{i-1})}{|V_i \cap M_i|}\Big] &= \Pr[G] \cdot \sum_{v\in M_i}p_v \phi'(v \mid S_{i-1}) \\
    &\ge \tfrac{1-\gamma}{|M_i|}\sum_{v\in M_i} \phi'(v \mid S_{i-1})
\end{align*}
where the first equality uses the fact that the inner sum is zero conditioned on $\overline{G}$. Because $V_i$ is chosen uniformly, $p_v$ is identical for all $v \in M_i$, thus $p_v = 1/|M_i|$. This yields \eqref{eq1}.
\end{proof}

The proof of the query complexity of \Cref{alg:sample_greedy} is given by the following lemma. 
\begin{lemma}\label{lemma:sample_greedy_complexity}
\Cref{alg:sample_greedy} makes $O(n \log \rk\log(1/\gamma))$ queries when $g(i)=\rk-i+1$ for all $i\in[\rk]$, and $O(n\log(1/\gamma))$ queries when $g(i)=\min\{\rk,\,n-i+1\}$ for all $i\in[\rk]$.
\end{lemma}
\begin{proof}
   The number of queries performed by \Cref{alg:sample_greedy} in iteration $i$ is  $|V_i|+1$, and we note that
\[
    |V_i|\le \frac{\abs{N_i}\ \log(1/\gamma)}{g(i)} + 1.
\]
Since a new element is added to $S_i$ in each iteration, we have $|N_i| = |V\setminus S_{i-1}| =  n-i+1$.
Hence, the total number of queries is bounded by
\begin{align*}
    \sum_{i=1}^\rk |V_i|+1 &\le 2\rk + \log(1/\gamma)\sum_{i=1}^\rk \frac{n-i+1}{g(i)}.
\end{align*}
If $g(i)=\rk-i+1$, the standard harmonic sum bound yields  
$\sum_{i=1}^\rk \frac{n-i+1}{g(i)}=O(n\log\rk)$, and the first part of the lemma follows. If $g(i)=\min\{\rk,\,n-i+1\}$, we have 
\[
    \sum_{i=1}^\rk \frac{n-i+1}{g(i)}\le \sum_{i=1}^\rk \frac{n-i+1}{\rk}+\sum_{i=1}^\rk \frac{n-i+1}{n-i+1} \le n+\rk=O(n),
\]
which implies the second part.
\end{proof}

We now turn to complete the proof for the privacy guarantee stated in \Cref{thm:sample_greedy_decomposable_non_oblivious}.
To this end, we utilize the following concentration bound.
\begin{lemma}[Claim 3.8 in \cite{chaturvedi2021differentially}, based on \cite{gupta2010differentially}]\label{lemma:concentration}
    Consider an $n$-round probabilistic process. In each round $i\in[n]$, an adversary chooses a distribution $\cP_i$ over $[0,1]$ and a sample $R_i$ is drawn from this distribution. Let $Z_1 =1$ and $Z_{i+1}=Z_i-R_iZ_i$. We define the random variable $Y_j = \sum_{i=j}^n Z_i\Exp[R_i]$. Then for any $j\in[n]$,
$
        \Pr[Y_j \ge qZ_j] \le \exp(3-q) .
 $ 
\end{lemma}
The above bound and the argument of \citet{gupta2010differentially,chaturvedi2021differentially} yields the following lemma. We continue with the notation established in the proof of \Cref{thm:sample_greedy_decomposable_non_oblivious}.
\begin{lemma}\label{lemma:sum_expec_bound}
The following inequality holds.
\begin{align*}
\Pr\left[\prod_{i=1}^\rk \E{u\sim\cP_i}[\exp(\tfrac{\eps_0}{2}\cdot \beta^i_t (u))]\ge   (e^{\eps_0/2}-1)(3+\log(1/\delta)) \right] \le \delta
\end{align*}
\end{lemma}

\begin{proof}
 For all $\eps_0 \in (0,1]$, we have
\begin{align*}
 \prod_{i=1}^\rk \E{u\sim\cP_i}[\exp(\tfrac{\eps_0}{2}\cdot \beta^i_t (u))] &\le  \prod_{i=1}^\rk \E{u\sim\cP_i}[1+(e^{\eps_0/2}-1)\cdot\beta^i_t (u)] \\
 &\le \exp\left( (e^{\eps_0/2}-1)\sum_{i=1}^\rk \E{u\sim\cP_i}[\beta^i_t (u)]\right).
\end{align*}
where the first inequality holds since $e^x\le 1+\tfrac{e^{\eps_0/2}-1}{\eps_0/2}x$ for all $x\in[0,\eps_0/2]$, and the second uses $1+x\le e^x$ for all $x$.
It remains to show that the sum of expectations is bounded with high probability.
Consider the following probabilistic process. In each round $i\in[\rk]$, $u_i$ is drawn from $\cP_i$. Let $Z_i = 1- \phi'_t(U_{i-1})$ be the total remaining marginal utility at the beginning of  iteration $i$. Define the random variable $R_{i-1}(u) =  \beta^i_t (u)/Z_{i-1}$
which is determined by the sample of $u\sim \cP_i$, and denotes the percentage increase in utility in  iteration $i$  (If $Z_i=0$, we set $R_i=0$). Since $\phi'_t$ is monotone and has range in $[0,1]$, we have $R_{i-1}(u)\in[0,1]$. Note that using these definitions, we have $Z_i = Z_{i-1}-R_{i-1} Z_{i-1}$. Define, $Y_j = \sum_{i=j}^\rk Z_i\cdot \Ex{u\sim\cP_i}{R_i(u)}$, and observe that the sum of interest is exactly $Y_1$. Letting $q= 3+\log(1/\delta) $ and using \Cref{lemma:concentration}, we conclude that 
\begin{align*}
\Pr\left[\sum_{i=1}^\rk \E{u\sim\cP_i}[\beta^i_t (u) ] \ge q\right] &= \Pr[Y_1 \ge q] \le \Pr[Y_1 \ge qZ_1]  \le \exp(3-q)   = \delta
\end{align*}
where the second inequality holds since $Z_1\in [0,1]$.
\end{proof}

\begin{proof}[Proof of \Cref{thm:sample_greedy_decomposable_non_oblivious}]
Let $D$ and $D'$ be two datasets such that $(D\setminus D')\cup (D'\setminus D)=\{x\}$. Suppose that instead of a set, \Cref{alg:sample_greedy} outputs the sequence of selected elements in their order of selection. Let $U = (u_1, \dots, u_\rk)$ be any such sequence, and $U_i = \{u_1, \dots, u_i\}$ denote the prefix set consisting of the first $i$ elements.

For any element $u \in V$, let $\beta_D^i(u) = \sum_{x \in D} \phi'_x(u \mid U_{i-1})$. By the definition of the exponential mechanism, the probability of selecting element $u$ at iteration $i$, given dataset $D$ and a sampled candidate set $V_i$, is:
\[
 \Pr[\ExpMech(q_i, V_i, D, \eps_0) = u] = \frac{\exp\left(\frac{\eps_0}{2} \beta_D^i(u)\right)}{\sum_{u' \in V_i} \exp\left(\frac{\eps_0}{2} \beta_D^i(u')\right)}.
\]
where we have used that since $\phi'_D$ is $1$-decomposable, it has a sensitivity $1/|D|$.
The expression for $D'$ follows analogously. Let $\cG$ denote \Cref{alg:sample_greedy}. By the chain rule, the probability of $\cG$ outputting sequence $U$ is:
\begin{align*}
 \Pr[\cG(D) = (u_1, \dots, u_\rk)] = \prod_{i=1}^{\rk} \Pr[\cG(D)_i = u_i \mid S_{i-1} = U_{i-1}].
\end{align*}
Let $\cV_i$ denote the collection of all subsets of $N_i$ of the size specified in Line~\ref{cardinality_sampling}. Since $V_i$ is sampled uniformly at random from $\cV_i$, for every $T \in \cV_i$, we have $\Pr[V_i = T] = p_i$, where $p_i = 1/|\cV_i|$ is independent of the dataset $D$. By the law of total probability, the $i$-th factor in the product above is:
\[
 \Pr[\cG(D)_i = u_i \mid S_{i-1} = U_{i-1}] = \sum_{T \in \cV_i: u_i \in T} \Pr[\ExpMech(q_i, T, D, \eps_0) = u_i] \cdot p_i.
\]
Note that the selection probability is zero if $u_i \notin V_i$. Using the fact that $\frac{\sum_j a_j}{\sum_j b_j} \le \max \frac{a_j}{b_j}$ for $a_j \ge 0$ and $b_j > 0$, we have:
\begin{align*}
\frac{\Pr[\cG(D)_i = u_i \mid S_{i-1} = U_{i-1}]}{\Pr[\cG(D')_i = u_i \mid S_{i-1} = U_{i-1}]} & = \frac{\sum_{T \in \cV_i, u_i \in T} \Pr[\ExpMech(q_i, T, D, \eps_0) = u_i]}{\sum_{T \in \cV_i, u_i \in T} \Pr[\ExpMech(q_i, T, D', \eps_0) = u_i]} \\
& \le \max_{T \in \cV_i: u_i \in T} \left[ \frac{\Pr[\ExpMech(q_i, T, D, \eps_0) = u_i]}{\Pr[\ExpMech(q_i, T, D', \eps_0) = u_i]} \right].
\end{align*}
Let $T_i$ be the candidate set attaining the maximum above. Taking the product over all iterations $\rk$:
\begin{align}
 \frac{\Pr[\cG(D) = U]}{\Pr[\cG(D') = U]} \le \left( \prod_{i=1}^\rk \frac{\exp(\frac{\eps_0}{2} \beta^i_D(u_i))}{\exp(\frac{\eps_0}{2} \beta^i_{D'}(u_i))} \right) \left( \prod_{i=1}^\rk \frac{\sum_{u \in T_i} \exp(\frac{\eps_0}{2} \beta^i_{D'}(u))}{\sum_{u \in T_i} \exp(\frac{\eps_0}{2} \beta^i_D(u))} \right).
\label{eq:factor}
\end{align}
We consider two cases.

\textbf{Case 1: $D = D' \cup \{x\}$.}
The first factor in \eqref{eq:factor} is bounded by \[\exp\Big(\tfrac{\eps_0}{2} \sum_{i=1}^\rk \phi'_x(u_i \mid U_{i-1})\Big) \le \exp(\tfrac{\eps_0}{2})\]
due to $1$-decomposability. The second factor is at most $1$ since $\beta_D^i(u) \ge \beta_{D'}^i(u)$.

\textbf{Case 2: $D' = D \cup \{x\}$.} The first factor is at most $1$, while the second factor becomes:
\begin{align}
 &\prod_{i=1}^\rk \frac{\sum_{u \in T_i} \exp(\frac{\eps_0}{2} \phi'_x(u \mid U_{i-1})) \exp(\frac{\eps_0}{2} \beta^i_D(u))}{\sum_{u \in T_i} \exp(\frac{\eps_0}{2} \beta^i_D(u))}\\
 &= \prod_{i=1}^\rk \mathbb{E}_{u \sim P_i} \left[ \exp\left(\tfrac{\eps_0}{2} \phi'_x(u \mid U_{i-1})\right) \right],
\label{eq:product2}
\end{align}
where $P_i$ is a distribution supported on $T_i$ with weights proportional to $\exp(\frac{\eps_0}{2} \beta^i_D(u))$. To bound this product of expectations, we apply the concentration bound from \cite{gupta2010differentially,chaturvedi2021differentially}. By \Cref{lemma:sum_expec_bound}, with probability at least $1-\delta$:
\[
 \frac{\Pr[\cG(D) = U]}{\Pr[\cG(D') = U]} \le \exp\left( (e^{\eps_0/2} - 1)(3 + \log(1/\delta)) \right).
\]
Combining both cases, for any neighboring $D, D'$ (under the change-one-element definition), with probability at least $1-\delta$:
\begin{align*}
 \frac{\Pr[\cG(D) = U]}{\Pr[\cG(D') = U]} &\le \exp \left( \frac{\eps_0}{2} + (e^{\eps_0/2} - 1)(3 + \log(1/\delta)) \right) \\
 &\le \exp \left( (e^{\eps_0/2} - 1)(4 + \log(1/\delta)) \right).
\end{align*}
Thus, \Cref{alg:sample_greedy} is $(\eps, \delta)$-DP with $\eps = (e^{\eps_0/2} - 1)(4 + \log(1/\delta))$.
\end{proof}

The remainder of this appendix is devoted to the proof of \Cref{thm:sample_greedy_guarantee_oblivious}, restated next.
\ObliviousSampleGreedy*

\DistMarginalGain*

\begin{proof}[Proof of \Cref{lemma:dist_marginal_gain}]
By \Cref{lemma:ravi_distance}, we have
$
\abs{S_{i-1}} \cdot d(C_i) \le (\abs{C_i} - 1) \cdot d(C_i, S_{i-1})
$.
Moreover, since $|S_{i-1}| = i-1$ and $|C_i| \le \rk$, we have $ d(C_i) \le \frac{\rk  - 1}{i-1} \cdot d(C_{i},S_{i-1})$. Therefore,
\begin{align*}
    d(\OPT\cup S_{i-1}) -d(S_{i-1}) &= d(C_i \cup S_{i-1}) -d(S_{i-1}) = d(C_i)+d(C_i,S_{i-1}) \le \paren{1+\frac{\rk -1}{i-1}}d(C_i,S_{i-1}).
\end{align*}
where the second equality holds because $S_{i-1}$ and $C_i$ are disjoint.
\end{proof}

\begin{lemma}\label{lemma:technical_prod}
    The following inequalities hold:
    \[
        \prod_{i=1}^{\rk -1}\paren{1-\frac{1-\gamma}{\rk (1+\frac{\rk -1}{i})}} \le \paren{\frac{2}{e}}^{(1-\gamma)(1-1/\rk)} \le \frac{2}{e}+\gamma+\frac{1}{\rk}.
    \]
\end{lemma}

\begin{proof}
We first evaluate the sum within the exponent. We have:
    \begin{align*}
          \frac{1}{\rk}\sum_{i=1}^{\rk -1}\frac{i}{i+\rk-1} &=   \frac{1}{\rk}\sum_{i=1}^{\rk -1}\paren{1-\frac{\rk -1}{i+\rk-1}}
          = \frac{\rk -1}{\rk} -\frac{\rk -1}{\rk}\cdot \sum_{i=1}^{\rk -1}\frac{1}{i+\rk-1}\\
          & = \frac{\rk -1}{\rk} -\frac{\rk -1}{\rk}\cdot \sum_{j=\rk}^{2\rk-2}\frac{1}{j}  
           \ge  \frac{\rk -1}{\rk} -\frac{\rk -1}{\rk}\cdot \log(2) \\
          & = \Big(1-\frac{1}{\rk}\Big)(1-\log 2),
    \end{align*}
    where in the inequality we have used the integral bound:
    \[
        \sum_{j=\rk}^{2\rk-2}\frac{1}{j} \le \int_{\rk -1}^{2\rk-2}\frac{1}{x}\,dx =  \log(2\rk-2)-\log(\rk-1) = \log(2).
    \]
   We next observe that the bound $\log(1 - x) \le -x$, which holds for all $x \in (0,1)$, together with the preceding derivation, implies that:
    \begin{align*}
         \sum_{i=1}^{\rk -1}\log\paren{1-\frac{1-\gamma}{\rk (1+\frac{\rk -1}{i})}} \le -(1-\gamma) \sum_{i=1}^{\rk -1}\frac{1}{\rk (1+\frac{\rk -1}{i})} \le -(1-\gamma)\Big(1-\frac{1}{\rk}\Big)(1-\log 2).
    \end{align*}
    The first inequality in the lemma follows by exponentiating both sides of the result above. For the second inequality, observe that:
    \begin{align*}
        \paren{2/e}^{(1-\gamma)(1-1/\rk)} &\le  \paren{2/e}^{(1-\gamma -1/\rk)} = e^{(1-\gamma - 1/\rk)\log(2/e)}
        = \frac{2}{e}\cdot e^{\log(e/2)(\gamma + 1/\rk)}\\
        &\le \frac{2}{e}\cdot (1+(e-1)\log(e/2)(\gamma+1/\rk)) \le \frac{2}{e} +\gamma +\frac{1}{\rk},
    \end{align*}
where the second inequality follows from the fact that $e^x \le 1+(e-1)x$ for all $x\in[0,1]$.
\end{proof}
\section{Omitted Proofs from Section \texorpdfstring{\protect\lowercase{\ref{sec:matroid_local_search}}}{\ref{sec:matroid_local_search}}}
\label[appendix]{appendix:localsearch}
This appendix contains the missing proofs for the results presented in \Cref{sec:matroid_local_search}. 
The following lemma is a  known property of matroids. 
\begin{lemma}[\cite{brualdi1969comments}]\label{lemma:matroid_property}
Let $X$ and $Y$ be two bases of a matroid $\cM= (V,\cI)$. Then, there exists a bijection
$h:X\to Y$ such that for every $u\in X$, $Y-h(u)+u\in\cI$,  and $h(u)=u$ for every $u\in X\cap Y$.
\end{lemma}

 We use the following notation, some of which has been introduced in \Cref{sec:matroid_local_search}. Let $S_i$ denote the base selected at iteration $i$, $C_i = \OPT \setminus S_{i-1}$, $B_i = S_{i-1} \setminus \OPT$, and let $h: S_{i-1} \to \OPT$ be a bijection satisfying the condition in \Cref{lemma:matroid_property}. Denote $B_i = \{b_1, \dots, b_t\}$, and let $c_j = h(b_j)$ for each $j \in[t]$.
The following result was proved in \cite{borodin2017max}. 
\begin{lemma}[Lemmas 5 and 7, \cite{borodin2017max}]\label{lemma:borodin_matroid_distance}
    Suppose $\cM$ has rank $\rk>2$ and $t\ge 2$.  Then,
    \begin{enumerate}[label=(\roman*)]
        \item $\sum_{j=1}^t [f(S_{i-1}-b_j +c_j) -f(S_{i-1})] \ge f(\OPT)-2f(S_{i-1})$,
        \item $\sum_{j=1}^t [d(S_{i-1}-b_j +c_j) -d(S_{i-1})] \ge d(\OPT)-2d(S_{i-1})$.
    \end{enumerate}
\end{lemma}

\begin{proof}[Proof of \Cref{lemma:sum_swaps_lowerbound}]
Recall that $h(u)=u$ for every $u\in S_{i-1}\cap\OPT$. Thus, when $t = 0$, we have $S_{i-1} = \OPT$, so the inequality becomes $\rk\phi(\OPT)\ge (\rk-1)\phi(\OPT)$.  When $t = 1$, the left-hand side becomes $\phi(\OPT) + (\rk - 1)\phi(S_{i-1})$. Since the claim holds in both cases, we assume $t \ge 2$. We get
\begin{align*}
    \sum_{u\in S_{i-1}}\phi(S_{i-1}-u+h(u)) &= \sum_{i=1}^t\phi(S_{i-1}-b_i +c_i) + (\rk-t)\phi(S_{i-1}) \\
    &\ge \phi(\OPT) + (t-2)\phi(S_{i-1}) + (\rk-t)\phi(S_{i-1}) \\
    & = \phi(\OPT) +(\rk-2)\phi(S_{i-1}).
\end{align*}
where the inequality follows by summing the inequalities in \Cref{lemma:borodin_matroid_distance} 
and rearranging.
\end{proof}

\begin{lemma}\label{lemma:hitting}
Let $V_i\subseteq V$ be a uniformly random subset of size $\lceil\frac{n}{\rk}\rceil$. Then, 
\[\Pr[(S_{i-1}\times V_i) \cap M_i \neq \emptyset] \ge 1-1/e.\]
\end{lemma}
\begin{proof}[Proof of \Cref{lemma:hitting}]
    Consider the probability that $V_i$ does not contain any element of $\set{h(u):u\in S_{i-1}}$. If $|V_i|>n-\rk$, this probability is zero. Otherwise, it equals
    \begin{align*}
        \frac{\binom{n-\rk}{|V_i|}}{\binom{n}{|V_i|}} &= \prod_{i=0}^{|V_i|-1}\frac{n-\rk-i}{n-i} = \prod_{i=1}^{\abs{V_i}}\paren{1-\frac{\rk}{n-i+1}} \le \paren{1-\frac{\rk}{n}}^{\abs{V_i}}\le e^{-\frac{\rk\abs{V_i}}{n}} \le e^{-1}.
    \end{align*}
where the last inequality holds since $\rk\cdot\abs{V_i} = \rk\cdot \lceil\frac{n}{\rk}\rceil \ge n$.
The claim follows.
\end{proof}

\begin{proof}[Proof of \Cref{lemma:sample_progress}]
Denote  $\alpha = 4\Delta \log( n )/\eps_0$. Fix an iteration $i\in [T]$ and let us condition on having selected a base $S_{i-1}$ after iteration $i-1$, and on some realization of the set $V_i$ at iteration $i$. Let $M_i=\set{(u,h(u)):u\in S_{i-1}}$.
By \cref{thm:exponential_mech}, the exponential mechanism guarantees that
\begin{align*}
    \Ex{\phi(S_{i-1}-u+v)-\phi(S_{i-1})} &\ge \max_{(u,v)\in W_i}[\phi(S_{i-1}-u+v)-\phi(S_{i-1})] - \alpha \\
    & = \max\set{0, \max_{\substack{(u,v)\in S_{i-1}\times (V_i\setminus S_{i-1})\\S_{i-1}-u+v\in\cI}}[\phi(S_{i-1}-u+v)-\phi(S_{i-1})]} -\alpha \\
    &\ge \max\set{0, \max_{(u,v)\in (S_{i-1}\times V_i)\cap M_i}[\phi(S_{i-1}-u+v)-\phi(S_{i-1})]} -\alpha \\
    & \ge \max\bigg\{0, \sum_{(u,v)\in (S_{i-1}\times V_i)\cap M_i}\frac{\phi(S_{i-1}-u+v)-\phi(S_{i-1})}{\abs{(S_{i-1}\times V_i)\cap M_i}}\bigg\} -\alpha
\end{align*}

 where the equality in the second line holds since $W_i$ always contains a dummy swap that does not change the current solution.   
 Unfixing the implicit conditioning on $V_i$ in iteration $i$ and taking expectation over all their possible realizations, we get
 \begin{align}
    \Ex{\phi(S_{i-1}-u+v)-\phi(S_{i-1})}  &\ge \Ex{V_i}{\max\bigg\{0, \sum_{(u,v)\in (S_{i-1}\times V_i)\cap M_i}\frac{\phi(S_{i-1}-u+v)-\phi(S_{i-1})}{\abs{(S_{i-1}\times V_i)\cap M_i}}\bigg\} -\alpha}  \nonumber\\
    & \ge   \max\bigg\{0, \E{V_i}\bigg[\sum_{(u,v)\in (S_{i-1}\times V_i)\cap M_i}\frac{\phi(S_{i-1}-u+v)-\phi(S_{i-1})}{\abs{(S_{i-1}\times V_i)\cap M_i}}\bigg]\bigg\} -\alpha   \nonumber\\
    &\ge \frac{1-1/e}{\rk}\sum_{(u,v)\in M_i}[\phi(S_{i-1}-u+v)-\phi(S_{i-1})] - \alpha \label{eq7}.
\end{align}
 The second inequality follows from Jensen's inequality, as the function \mbox{$f(x)=\max\set{0,x}$} is convex. We now focus on proving the last inequality.
First, if the sum in \eqref{eq7} is negative, the final inequality holds trivially; we therefore assume without loss of generality that it is non-negative.
Consider the following probabilistic experiment. We sample  $V_i$ as in \Cref{localsearch_sampling}. If $(S_{i-1}\times V_i) \cap M_i\neq \emptyset$, we select a uniformly random  $(u,v)\in (S_{i-1}\times V_i) \cap M_i$. Otherwise, output $\bot$. Let $p_{u,v}$ denote the probability that $(u,v)\in M_i$ is selected conditioned on the event  $(S_{i-1}\times V_i) \cap M_i\neq \emptyset$, which we denote by $G$. We have
\begin{align*}
    \E{V_i}\bigg[\sum_{(u,v)\in (S_{i-1}\times V_i)\cap M_i}\frac{\phi(S_{i-1}-u+v)-\phi(S_{i-1})}{\abs{(S_{i-1}\times V_i)\cap M_i}}\bigg]&= \Pr[G] \cdot \sum_{(u,v)\in M_i}p_{u,v}\cdot [\phi(S_{i-1}-u+v)-\phi(S_{i-1})] \\
    &\ge \frac{1-1/e}{\rk}\sum_{(u,v)\in M_i} [\phi(S_{i-1}-u+v)-\phi(S_{i-1})]
\end{align*}
where in the first equality  we use the law of total expectation, observing that the inner sum is empty and thus equals $0$ conditioned on $\overline{G}$. In the second inequality, we observe that since  $V_i$ is chosen uniformly at random, the probabilities $p_{u,v}$ must be the same for all $(u,v)\in M_i$, and therefore equal $1/|M_i|=1/\rk$.
By \Cref{lemma:hitting} we have $\Pr[G]\ge 1-1/e$, and the inequality~\eqref{eq7} follows. Using \Cref{lemma:sum_swaps_lowerbound}, we get
\begin{align*}
    \Ex{\phi(S_{i-1}-u+v)} \ge\phi(S_{i-1}) + \frac{1-1/e}{\rk}\paren{\phi(\OPT)-2\phi(S_{i-1})} -\alpha.
\end{align*}
Finally, the lemma follows by unfixing $S_{i-1}$ and taking an expectation over its possible realizations.

\end{proof}

\begin{proof}[Proof of \Cref{thm:sample_local_search_guarantee}]
The privacy guarantee follows directly from the $\eps_0$-DP of the exponential mechanism and the composition theorems (\Cref{thm:composition}), noting that we now have $T' \eqdef T+1 = O(\gamma ^{-1} \rk \log \rk)$ compositions.
 For the query
complexity, \Cref{alg:sample_local_search} performs $T$ iterations, and each iteration requires only 
$O(\abs{S_{i-1}} \cdot \abs{V_i}) = O(n)$ queries due to the sub-sampling step.
Overall, the algorithm performs $O(\gamma^{-1}n\rk\log \rk)$ queries. We proceed to prove the utility guarantee.
Let $\alpha = 4\Delta \log( n )/\eps_0$. 
We may assume that
$
    \frac{1-1/e}{\rk} \cdot \phi(\OPT) - \alpha \ge \frac{\phi(\OPT)}{8\rk}
$,
since otherwise, it can be verified that \mbox{$\phi(\OPT) < 2\rk\alpha$},  and the guarantees of \Cref{thm:sample_local_search_guarantee} hold trivially.
Let $S_0$ be the first base chosen by the algorithm. 
We have
\begin{align*}
    \Ex{\phi(S_1)} &\ge  \phi(S_0)+  \frac{1-1/e}{\rk}\paren{\phi(\OPT) -2\phi(S_0)}- \alpha\\
    &\ge \frac{1-1/e}{\rk}\cdot \phi(\OPT) - \alpha \ge \frac{\phi(\OPT)}{8\rk}
\end{align*}
where the first inequality holds by \Cref{lemma:sample_progress}, and the second since $\rk\ge 2$.
Now, fix an iteration $2\le i\le T$, and let us condition on having selected some base $S_{i-1}$ after the first $i-1$ iterations.
If $\Ex{\phi(S_{i-1})} < \phi(\OPT)/(2+\gamma) -\rk\alpha$, then, by  \Cref{lemma:sample_progress} and since $\alpha \ge 0$, we obtain
\begin{align}
    \Ex{\phi(S_i)}
    \ge \paren{1+\frac{\gamma(1-1/e)}{\rk}}\Ex{\phi(S_{i-1})} \label{eq:expectation_progress}
\end{align} 
We claim that given our choice of $T$, there must exist some $i\in[T]$ for which $\Ex{\phi(S_i)} \ge \phi(\OPT)/(2+\gamma) -\rk\alpha$. Assume for contradiction that this is not the case.  Recursively applying the bound  \eqref{eq:expectation_progress} yields
\begin{align*}
    \Ex{\phi(S_T)} &\ge \Bigl(1+\frac{\gamma(1-1/e)}{\rk}\Bigr)^{\frac{2\rk\log(8\rk)}{\gamma(1-1/e)}} \Ex{\phi(S_1)} \\
    &\ge\frac{1}{8\rk} \Bigl(1+\frac{\gamma(1-1/e)}{\rk}\Bigr)^{\frac{2\rk\log(8\rk)}{\gamma(1-1/e)}} \phi(\OPT)\\
    &> \phi(\OPT).
\end{align*}

The last inequality from the bound $(1+1/x)^x \ge \exp(1-1/(2x))$ which holds for all $x>0$.
In particular, there exists some base $S$ with $\phi(S)>\phi(\OPT)$, in contradiction.
Thus, by the law of total expectation, we get
\begin{align*}
    \Ex{\phi(S_{i^*})} &= \Ex{\Ex{\phi(S_{i^*})}\mid S_1,\dots, S_T} \ge \Ex{\max_{i\in[T]}{\phi(S_{i})}} - \tfrac{2\Delta \log T} {\eps_0} \\
    &\ge \max_{i\in[T]}\Ex{\phi(S_{i})} - \tfrac{2\Delta \log T}{\eps_0} \ge \tfrac{\phi(\OPT)}{2+\gamma} -\rk \alpha -\tfrac{2\Delta \log T}{\eps_0} \\
    &= \tfrac{\phi(\OPT)}{2 + \gamma} - \tfrac{2\Delta}{\eps_0} \cdot \left( 2\rk \log n + \log\left( \tfrac{2\rk\log \rk}{\gamma(1-1/e)}\right) \right)
\end{align*}
The second inequality holds by the fact that 
$
\mathbb{E}\left[\max_i X_i\right] \geq \max_i \mathbb{E}[X_i]
$
for any jointly distributed random variables $X_1, \ldots, X_n$.
As $1/(2+\gamma) \ge 1/2 -\gamma$, the theorem follows.
\end{proof}

\section{Local Search for Rank 2 Matroid}\label[appendix]{appendix:rank2}
This appendix provides the analysis of \Cref{alg:sample_local_search} for matroids of rank~2, thereby completing the proof of \Cref{thm:sample_local_search_guarantee} for arbitrary matroid rank.
We follow the notation of \Cref{sec:matroid_local_search}. The difficulty arises because the second inequality in \Cref{lemma:borodin_matroid_distance} may fail when $\rk=2$. Nevertheless, as its proof shows, \Cref{lemma:sum_swaps_lowerbound} holds for $t=0$ or $t=1$ even with $\rk=2$, so \Cref{lemma:sample_progress} remains valid in these cases. Therefore, to complete the proof of \Cref{thm:sample_local_search_guarantee}, it suffices to verify that the inequality in \Cref{lemma:sample_progress} continues to hold in iterations with $\rk=t=2$, namely when $S_{i-1}=\{b_1,b_2\}$ and $C_i=\{c_1,c_2\}$.

\begin{lemma}
    Suppose that $\cM$ has rank $\rk=2$. Conditioned on having selected a base $S_{i-1}=\set{b_1,b_2}$ after iteration $i-1$, the following inequality holds.
     \begin{align*}
    \Ex{\phi(S_{i})} \ge \phi(S_{i-1}) + \frac{1-1/e}{\rk}\paren{\phi(\OPT)-2\phi(S_{i-1})} - \frac{4\Delta \log( n )}{\eps_0}
\end{align*}
\end{lemma}
\begin{proof}
    Since $\cM$ has rank $2$, our assumption implies that $\OPT=\set{c_1,c_2}$. Without loss of generality, assume that $f(b_1)\ge f(b_2)$. 
By the triangle inequality,
\begin{align*}
    d(S_{i-1}-b_2+c_1)+d(S_{i-1}-b_2+c_2) & = d(b_1,c_1) + d(b_1,c_2)  \ge d(c_1,c_2) = d(\OPT) \label{eq1: lemma:progress}.
\end{align*}
Moreover,
\begin{align*}
  f(\OPT)-f(S_{i-1}) & \le f(\set{b_1,b_2,c_1,c_2}) - f(\set{b_1,b_2})\\
  & \le f(\set{b_1,b_2,c_1}) + f(\set{b_1,b_2,c_2}) - 2\cdot f(\set{b_1,b_2}) \\
  & \le f(\set{b_1, c_1}) + f(\set{b_1,c_2}) - 2\cdot f(\set{b_1}) \\
  & =  f(S_{i-1}-b_2+c_1) + f(S_{i-1}-b_2+c_2) - 2\cdot f(\set{b_1}) 
\end{align*}
where the first inequality follows from monotonicity, and the second and third follow from submodularity. Rearranging and using the fact that  $f(S_{i-1}) \le f(\set{b_1})+f(\set{b_2}) \le 2 f(\set{b_1})$ yields
\begin{align*}
   f(S_{i-1}-b_2+c_1) + f(S_{i-1}-b_2+c_2)  \ge f(\OPT).
\end{align*}
Now, let $M_i = \set{(b_2,c_1), (b_2,c_2)}$. By summing the previous inequalities, we obtain
\[
    \sum_{(u,v)\in M_i}[\phi(S_{i-1}-u+v)-\phi(S_{i-1})] \ge \phi(\OPT) - 2\phi(S_{i-1}).
\]
The remainder of the proof is analogous to that of \Cref{lemma:sample_progress}, but uses the current definition of $M_i$. A similar derivation yields
 \begin{align*}
    \Ex{\phi(S_{i-1}-u+v)-\phi(S_{i-1})}    &\ge \frac{1-1/e}{\rk}\sum_{(u,v)\in M_i}[\phi(S_{i-1}-u+v)-\phi(S_{i-1})] - \alpha \\
                                            & \ge \frac{1-1/e}{\rk}\paren{\phi(\OPT)-2\phi(S_{i-1})} -\alpha
\end{align*}
using that $\Pr[\set{c_1,c_2} \cap V_i \neq \emptyset] \ge 1 - e^{-1}$, by an argument similar to that used in the proof of \Cref{lemma:hitting}.
\end{proof}

\section{Sensitivity Analysis for the Uber Objective}\label[appendix]{appendix:sens_uber}

In this appendix, we show that the Uber objective defined in \Cref{sec: exp} is $1/m$-decomposable. The derivation follows a similar logic to that of the Amazon product selection objective presented in \Cref{example:decomposable}.

For a customer location $a=(x_a,y_a)$ and a grid point $b=(x_b,y_b)$, the convenience score $c(a,b) \in [0,1]$ is defined as in \Cref{sec: exp}. The relevance of a set $S$ of selected grid points is given by $f_D(S) = \frac{1}{m} \sum_{a \in D} \max_{b \in S} c(a,b)$. Furthermore, the dispersion $d(S)$ depends only on the distance $d(b, b')$ between selected grid points, which is independent of the sensitive data and is also normalized to $[0,1]$.

To establish decomposability, we associate each pickup location $a \in D$ with a monotone submodular utility function $f_a(S) = \max_{b \in S} c(a,b)$. This function represents the convenience score of the set $S$ for individual $a$ and depends exclusively on that specific individual. For $\lambda \in [0,1]$, we define the per-user objective as follows:
\[ \phi_a(S) = \frac{(1-\lambda)}{m} f_a(S) + \frac{2\lambda}{mk(k-1)} d(S) \]
where $m=|D|$. The normalization ensures that $\phi_a(S) \in [0, 1/m]$. The total objective is the sum of these individual contributions:
\[ \phi(S) = \sum_{a \in D} \phi_a(S) = \frac{1-\lambda}{m} \sum_{a \in D} f_a(S) + \frac{2\lambda}{k(k-1)} d(S). \]
Consequently, $\phi$ is $1/m$-decomposable. To derive the sensitivity bound for $\phi$, we apply the same argument as in~\Cref{example:decomposable}. If $D = (D' \setminus \{x'\}) \cup \{x\}$, then for any $S \subseteq V$,
$\abs{\phi_D(S) - \phi_{D'}(S)} = \abs{\phi_x(S) - \phi_{x'}(S)} \le 1/m$,
and thus $\phi$ has sensitivity bounded by $1/m$.

\section{Additional Experimental Results}
\label[appendix]{appendix:experiments}

In this appendix, we provide supplementary experimental evaluations to further validate the robustness and scalability of our proposed algorithms.

\paratitle{Execution time by $m$}
\Cref{fig:time_by_m} shows the scaling of execution time with respect to the number of users $m = |D|$ for the Uber location selection experiment. We emphasize that $|D|$ does not directly affect the algorithmic complexity. Rather, its impact is application-specific, as it only influences the cost of evaluating the objective oracle, which may vary significantly across different instantiations of our framework. Nevertheless, we include these results for completeness.

The results demonstrate an essentially linear dependence on $m$, with a slightly sharper increase only at the largest scale. Notably, our accelerated algorithms, \DPNOSG and \DPOSG, exhibit smaller constant factors than the baselines, further validating the practical efficiency and scalability of our approach.
\begin{figure}[t]
    \centering
    \begin{minipage}[t]{0.31\textwidth}
        \centering
        \includegraphics[width=\linewidth, height=0.14\textheight, keepaspectratio]{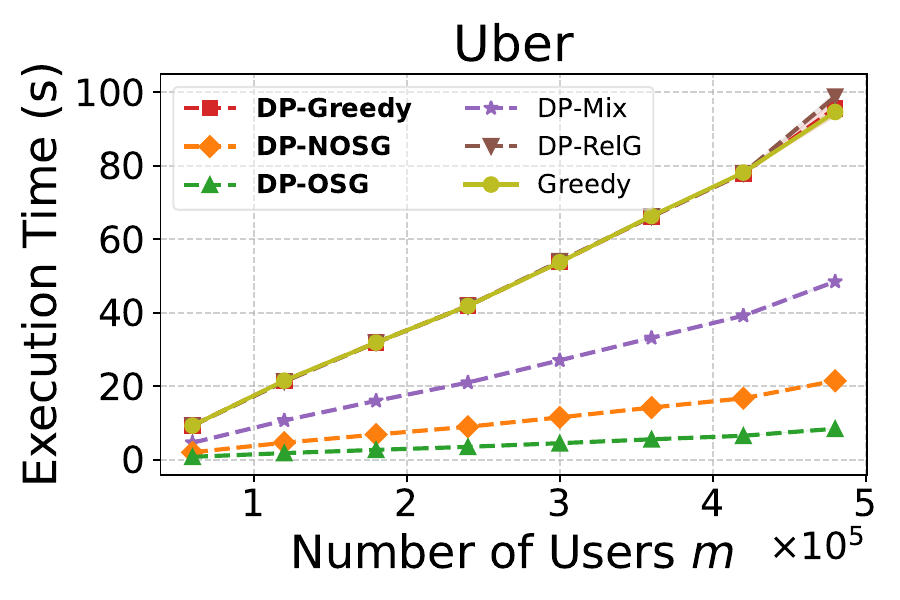}
        \caption{Execution time (seconds) as the number of users $m$ varies.}
        \label{fig:time_by_m}
    \end{minipage}
    \hfill
    \begin{minipage}[t]{0.31\textwidth}
        \centering
        \includegraphics[width=\linewidth, height=0.14\textheight, keepaspectratio]{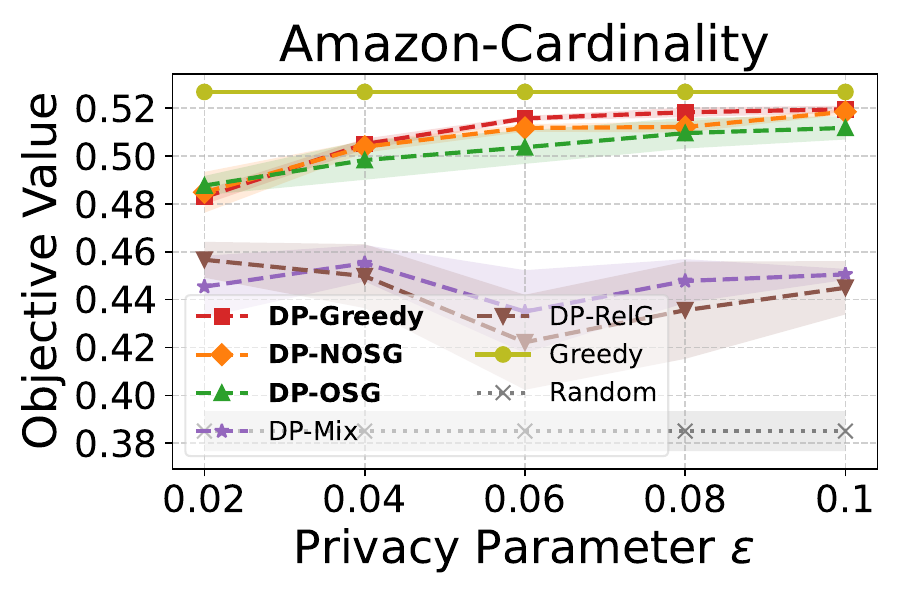}
        \caption{Objective value as the privacy parameter $\eps$ varies.}
        \label{fig:val_for_eps_amazon}
    \end{minipage}
    \hfill
    \begin{minipage}[t]{0.31\textwidth}
        \centering
        \includegraphics[width=\linewidth, height=0.14\textheight, keepaspectratio]{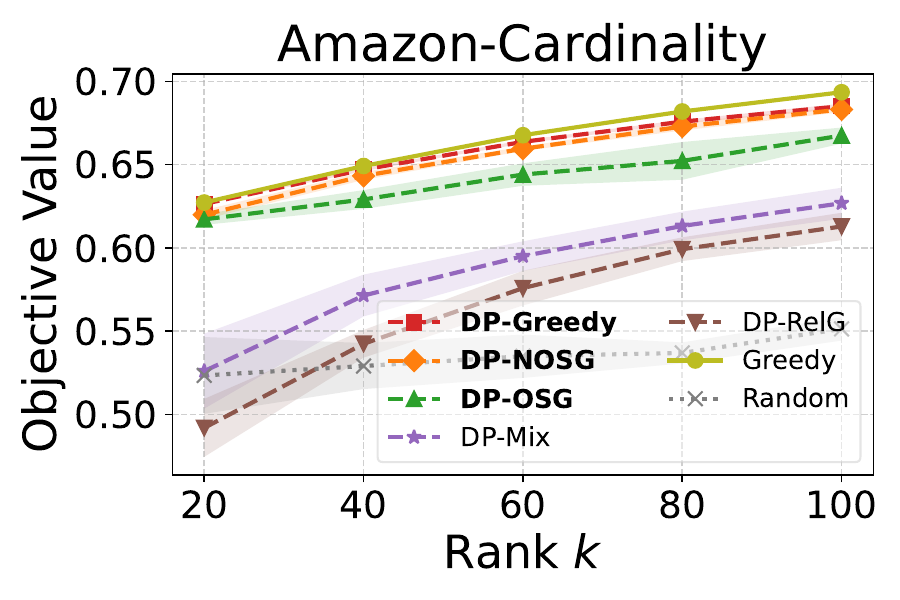}
        \caption{Objective value as the rank $k$ varies.}
        \label{fig:val_for_k_amazon}
    \end{minipage}
\end{figure}

\paratitle{Amazon-Cardinality Utility}
\Cref{fig:val_for_eps_amazon} shows the impact of the privacy parameter $\eps$ on utility for the Amazon product summarization dataset under a cardinality constraint with $k = 20$, while \Cref{fig:val_for_k_amazon} illustrates the effect of varying $k$ for $\eps = 0.1$. Overall, the results show that our algorithms consistently outperform the DP baselines and achieve utility competitive with the non-private \Greedy baseline. Performance experiments are omitted, as the observed trends are similar to those in the Uber experiment.

\fi

\end{document}